\documentclass{usiinftr}
\usepackage{amsmath}
\usepackage{amstext}
\usepackage{amssymb}
\usepackage{amsfonts}
\usepackage{amsthm}
\usepackage{xspace}
\usepackage{graphicx}
\usepackage{url}
\usepackage{subfigure}
\usepackage{mdwlist}
\usepackage{courier}
\usepackage{lscape}
\usepackage{comment}
\usepackage{wrapfig}
\usepackage{multirow}
\usepackage{bussproofs}
\usepackage{pdftricks}
\usepackage{subfigure}
\begin{psinputs}
\usepackage{pstricks}
\usepackage{color}
\usepackage{pstcol}
\usepackage{pst-plot}
\usepackage{pst-tree}
\usepackage{pst-eps}
\usepackage{multido}
\usepackage{pst-node}
\usepackage{pst-eps}
\end{psinputs}
\usepackage{cite}
\makeatletter
\newif\if@restonecol
\makeatother

\usepackage[ruled,vlined]{algorithm2e}
\usepackage{footmisc}
\usepackage{balance}
\usepackage{enumerate}
\usepackage[normalem]{ulem}
\usepackage{hyperref}
\usepackage{alltt}

\newcommand{\formulae}{formul\ae\xspace}

\newcommand{\dpll}[1]{{\sc DPLL}\xspace}

\newcommand{\COMMENT}[1]{}

\newcommand{\ui}{\ensuremath{\underline i}}
\newcommand{\uj}{\ensuremath{\underline j}}
\newcommand{\uk}{\ensuremath{\underline k}}

\newcommand{\ux}{\ensuremath{\underline x}}
\newcommand{\uy}{\ensuremath{\underline y}}

\renewcommand{\int}{\ensuremath {\mathcal I}}

\newcommand{\ta}{\ensuremath{{\bf a}}\xspace}
\newcommand{\tc}{\ensuremath{{\bf c}}\xspace}
\newcommand{\td}{\ensuremath{{\bf d}}\xspace}
\newcommand{\tv}{\ensuremath{{\bf v}}\xspace}
\newcommand{\tw}{\ensuremath{{\bf w}}\xspace}
\newcommand{\tu}{\ensuremath{{\bf u}}\xspace}
\newcommand{\ttb}{\ensuremath{{\bf t}}\xspace}

\theoremstyle{definition}
\newtheorem{definition}{Definition}[section]
\newtheorem{example}{Example}[section]

\theoremstyle{plain}
\newtheorem{theorem}{Theorem}[section]

\newtheorem{proposition}{Proposition}[section]

\begin{document}

%
%
\title{Abstraction and Acceleration in SMT-based Model-Checking for Array Programs}

%
%
\author{Francesco Alberti}{1}
\author{Silvio Ghilardi}{2}
\author{Natasha Sharygina}{1}

%
%
\affiliation{1}{\USIINF}
\affiliation{2}{Universit\`a degli Studi di Milano, Milan, Italy}

%
%
\TRnumber{2012-1}

%
%
\date{October~2012}
\revised{April~2013}

%
%


\maketitle

\begin{abstract}
Abstraction (in its various forms) is a powerful established technique in model-checking; 
still, when unbounded data-structures are concerned, it cannot always cope with divergence phenomena 
in a satisfactory way. Acceleration is an approach which is widely used to avoid divergence, but it 
has been applied mostly to integer programs.
This paper addresses the 
problem of accelerating transition relations for unbounded arrays with the ultimate goal of avoiding 
divergence during reachability analysis of abstract programs.
For this, we first design a format to compute accelerations in this domain;
then we show how
to adapt the so-called `monotonic abstraction' technique to efficiently handle complex \formulae with 
nested quantifiers generated by the acceleration preprocessing. Notably, our technique can be easily 
plugged-in into
abstraction/refinement loops, and strongly contributes to avoid divergence: experiments 
conducted with the MCMT model checker attest the effectiveness of our approach on programs with 
unbounded arrays, where
acceleration and abstraction/refinement technologies fail if applied alone.
\end{abstract}

\section{Introduction}
Transitive closure is a logical construct that is far beyond first order
logic: either infinite disjunctions or higher order quantifiers or, at
least, fixpoints operators are required to express it. 
Indeed, due to the compactness of first order logic,
transitive closure (even modulo the axioms of a first order theory) is
first-order definable only in trivial cases.
These general results do not hold if we define a theory as a
class of structures $\mathcal{C}$ over a given signature\footnote{Such
definition is widely adopted in the SMT literature~\cite{SMT-LIB2}.}.
Such definition is different from the ``classical'' one where a theory
is identified as a set of axioms. By taking a theory as a class of
structures the property of compactness breaks, and
%
it might well happen that transitive closure 
becomes first-order definable (the first order definition being valid just inside the class
$\mathcal C$ - which is often reduced to a single structure). 

In this paper we consider the extension of Presburger arithmetic with
free unary function symbols.
%
Inside Presburger arithmetic, various classes of relations are known to have definable 
{\it acceleration}\footnote{`acceleration' is the name usually adopted
in the formal methods literature to indicate transitive closure.} (see
related work section below).
In our combined setting, the presence of free function symbols
introduces a 
novel 
feature that, for instance, limits decidability to
controlled extensions of the quantifier-free fragment~\cite{BradleyMS06,
GeM09}.
In this paper we show that in such theory some classes of relations
admit a definable acceleration.

The theoretical problem of studying the definability of accelerated
relations has an important application in program verification.
The theory we focus on is widely adopted to represent programs handling arrays, where free functions model arrays of integers. 
In this application domain, the accelerated counterpart of relations encoding systems evolution (e.g., loops in programs) allows to compute
`in one shot' the reachable set of states after an arbitrary but finite number of execution steps.
This has the great advantage of keeping under control sources of
(possible) divergence arising in the reachability analysis.
The contributions of the paper are many-fold.
First, we show that inside the combined theory of Presburger arithmetic augmented with free function symbols, 
the acceleration of some classes of relations -- corresponding, in our application domain, to relations involving arrays and counters --
can be expressed in first order language. This result comes at a price of allowing nested quantifiers.
Such nested quantification can be problematic in practical
applications. To address this complication,
as a second contribution of the paper, we show how to take care of
the quantifiers added by the accelerating procedure: the idea is
to import in this setting the so-called \emph{monotonic abstraction}
technique \cite{cav06, tacas06}.
Such technique has been reinterpreted and analyzed in a declarative
context in \cite{jsat}: from a logical point of view, it amounts to a
restricted form of \emph{instantiation for universal quantifiers}.
Third, we show that the ability to compute accelerated relations is greatly beneficial in
program verification. In particular, one of the biggest problems in verifying safety
properties of array programs is designing procedures for the synthesis of
relevant quantified predicates. In typical sequential programs (like those illustrated
in \figurename \ref{fig:motivating}), the guarded assignments used to model the program
instructions are ground and, as a consequence, the formulae representing backward
reachable states are ground too.
However, the invariants required to certify the safety of such programs contain quantifiers.
Our acceleration procedure is able to supply the required quantified predicates.
%
%
Our experimentation attests that abstraction/refinement-based
strategies widely used in verification benefit from accelerated transitions. In programs with nested
loops, as the {\tt allDiff} procedure of \figurename \ref{fig:motivating} for example, 
the ability to accelerate the inner loop simplifies the structure of
the problem, 
allowing abstraction to converge during verification of the entire program.
For such programs, abstraction/refinement 
or acceleration approaches taken in isolation are not sufficient,
reachability analysis converges only if they are combined together.

\noindent
{\it Related Work.}~
To the best of our knowledge, the only work addressing the problem of
accelerating relations involving arrays is \cite{BozgaHIKV09}. Such approach seems to be unable to
handle properties of common interest with more than one quantified variable (e.g., ``sortedness'') and
is limited to programs without nested loops. Our technique is not
affected by such limitations and can successfully handle examples
outside the scope of \cite{BozgaHIKV09}.

Inside Presburger arithmetic, various classes of relations are known to have definable 
acceleration:
these include 
relations that can be formalized as
difference bounds constraints~\cite{db1,db2}, octagons~\cite{octagons}
and finite monoid affine transformations~\cite{presburger} (paper~\cite{acceleration_CAV} presents a general approach covering all
these domains).
Acceleration for relations over Presburger arithmetic
has been also plugged into
abstraction/refinement loop for verifying integer programs~\cite{CaniartFLZ08, acceleration_ATVA}.


We recall that acceleration has also been applied 
proficiently in the
analysis of real time systems (e.g., \cite{HendriksL02,
BehrmannBDLPY02}), to compactly represent
the iterated
execution of cyclic actions (e.g., polling-based systems) and address fragmentation problems.

Our work can be proficiently combined with SMT-based
techniques for the verification of programs, as it helps
helps avoiding the reachability analysis divergence
when it comes to abstraction of programs with arrays of unknown length.
Since the technique mostly operates at pre-processing level (we add 
to the system
accelerated transitions
by collapsing branches of loops handling arrays), \emph{we believe that our technique is compatible with 
most approaches}
proposed in array-based software model checking.
We summarize some of these approaches below, without pretending of being exhaustive.

The vast majority of software model-checkers implement abstraction-refinement algorithms
(e.g., \cite{GrafS97, ClarkeGJLV00, BallR01
}).
{\it Lazy Abstraction with Interpolants} \cite{McMillan06} is one of the most effective frameworks for unbounded reachability analysis of programs.
It relies on the availability of interpolation procedures (nowadays efficiently embedded in SMT-Solvers \cite{CimattiGS10}) to generate new predicates
as (quantifier-free) interpolants for refining infeasible counterexamples.
%
\begin{figure}[t]
\begin{minipage}{0.44\textwidth}
\centering
{\scriptsize
$$
\begin{aligned}
&  {\sf function}~{\tt allDiff}~(~{\sf int}~{\tt a}[{\tt N}]~):~\\
& 1~~{\tt r}={\sf true}; \\
& 2~~{\sf for}~( {\tt i} = 1;~{\tt i} < {\tt N} \wedge {\tt r}; {\tt i}\text{+}\text{+} )~\\
& 3~~~~~{\sf for}~({\tt j} = {\tt i}\text{-}1; {\tt j} \geq 0 \wedge {\tt r}; {\tt j}\text{-}\text{-})~\\
& 4~~~~~~~{\sf if}~({\tt a}[{\tt i}] = {\tt a}[{\tt j}])~{\tt r} = {\sf false};\\
& 5~~{\sf assert}~\left( {\tt r} \rightarrow \left(
\begin{aligned}
& \forall x,y (0 \leq x < y < {\tt N}) \\
& ~~~~~\rightarrow ({\tt a}[x] \neq {\tt a}[y])
\end{aligned}
\right)
\right) \\
\end{aligned}
$$
}
\begin{center}
(a)
\end{center}
\end{minipage}\hfill
\begin{minipage}{0.55\textwidth}
\centering
{\scriptsize
$$
\begin{aligned}
&  {\sf function}~{\tt Reverse}~(~{\sf int}~{\tt I}[{\tt N}+1];~{\sf int}~{\tt O}[{\tt N}+1];~{\sf int}~{\tt c}~):\\
& 1~~{\tt c} = 0; \\
& 2~~{\sf while}~( {\tt c} \neq N+1 )~\{
 {\tt O}[{\tt c}] = {\tt I}[N-{\tt c}]; {\tt c\!+\!+};
\}\\
& 3~~{\sf assert}~\left(
\begin{aligned}
& \forall x\geq 0, y\geq 0\\
& ~~~ (x+y={\tt N}\to {\tt I}[x] = {\tt O}[y]~)~
\end{aligned}
\right) \\
\\
\\
\end{aligned}
$$
}
\begin{center}
(b)
\end{center}
\end{minipage}
\caption{\label{fig:motivating}Motivating examples.}
\end{figure}
%

%


For programs with arrays of unknown length the classical interpolation-based lazy abstraction works only 
if there is a support to handle quantified predicates
\cite{ABGRS12a} (the approach of~\cite{ABGRS12a} is the basis of our experiments below).
Effectiveness and performances of abstraction/refinement approaches
strongly depend on their ability in generating
the ``right'' predicates to stop divergence of verification procedures.
%
%
In case of programs with arrays, this quest can rely
on {\it ghost variables} \cite{FlanaganQ02} retrieved from the
post-conditions, on the backward propagation of post-conditions along
spurious counterexamples \cite{SeghirPW09} or can be constraint-based
\cite{SG09, BeyerHMR07}. Recently, constraint-based
techniques have been significantly extended to the generation of loop invariants
outside the array property fragment \cite{LarrazRR13}. This solution
exploits recent advantages in SMT-Solving, namely those devoted to
finding solutions of constraints over non-linear integer arithmetic
\cite{BorrallerasLORR12}.
Other ways to generate predicates are by means of {\it saturation-based}
theorem provers \cite{KV09,McM08} or interpolation procedures
\cite{JhalaM07, ABGRS12a}.
%

All the aforementioned techniques suffer from a certain degree of
randomness due to the fact that detecting the ``right'' predicate is an
undecidable problem.
%
%
For example, predicate abstraction approaches (i.e., \cite{ABGRS12a,ABGRS12b,SeghirPW09}) fail
verifying the procedures
in \figurename \ref{fig:motivating}, which  are commonly considered to be challenging for verifiers because they cause divergence\footnote{
The procedure {\tt Reverse} outputs to the array {\tt O} the reverse of the array {\tt I}; the procedure {\tt allDiff} checks whether the entries 
of the array {\tt a} are all different.
Many thanks to
Madhusudan Parthasarath and his group for pointing us to challenging problems with arrays of unknown length, including the {\tt allDiff} example.}.
Acceleration, on the other side, provides a precise and systematic way
for addressing the verification of programs. Its combination, as a
preprocessing procedure, with standard abstraction-refinement techniques
allows to successfully solve challenging problems like the ones in \figurename
\ref{fig:motivating}.

The paper is structured as follows: Section \ref{sec:preliminaries}
recalls the background notions about Presburger arithmetic 
and extensions.
In order to identify the classes of relations whose acceleration we want to study, we are guided by software model checking applications. To this end, 
we provide in Section~\ref{sec:guarded}  classification of the guarded assignments we
are
interested in.
Section
\ref{sec:backward} demonstrates the practical application of the
theoretical results. In particular, it presents a backward reachability procedure and
shows how to plug acceleration with monotonic abstraction in it. The
details of the theoretical results are presented later. 
The main definability result for accelerations is in
Section~\ref{sec:accelerating}, while
Section~\ref{sec:iterators} introduces the 
abstract notion
of an iterator. Section \ref{sec:experiments} discusses our
experiments and
Section \ref{sec:conclusion} 
concludes the paper.

\section{Preliminaries}\label{sec:preliminaries}
We work in Presburger arithmetic enriched with free function symbols and with definable function symbols (see below);
when we speak about validity or satisfiability of a formula, 
we mean 
\emph{satisfiability and validity in all structures having the standard structure of natural numbers as reduct}. 
Thus, satisfiability and validity are decidable if we limit to quantifier-free \formulae 
(by adapting
Nelson-Oppen combination results~\cite{NelOpp,TinHar}), 
but may become
undecidable otherwise (because of the presence of free function symbols).

We use $x,y,z, \dots$ or $i,j,k, \dots$ for variables; $t, u, \dots$
for terms, $c, d, \dots$ for free constants, $a, b, \dots$ for free function symbols, $\phi, \psi, \dots$ for \emph{quan\-tifier-free} \formulae.
Bold letters are used for tuples
and $\vert -\vert$ indicates tuples length;
 hence for instance $\tu$ indicates a tuple of terms like $u_1, \dots, u_m$, where ${m=\vert \tu\vert}$ (these tuples may contain repetitions). For variables, we use underline letters $\ux, \uy, \dots, \ui, \uj, \dots$ to indicates
tuples without repetitions. Vector notation can also be used for equalities: if $\tu=u_1, \dots, u_n$ and $\tv=v_1, \dots, v_n$, we may use $\tu =\tv$ to mean the formula
$\bigwedge_{i=1}^n u_i=v_i$.

If we write  $t(x_1, \dots, x_n), \tu(x_1, \dots, x_n), \phi(x_1, \dots , x_n)$ (or $t(\ux), \tu(\ux), \phi(\ux), \dots$,  in case  $\ux=x_1, \dots, x_n$), we mean that the term $t$, 
the tuple of terms $\tu$, the quantifier-free formula $\phi$ contain variables only from the tuple $x_1, \dots, x_n$.
Similarly, we may use $t(\ta, \tc, \ux), \phi(\ta, \tc, \ux), \dots$ to mean \emph{both} that the term $t$ or the quantifier-free formula $\phi$ have free variables included in $\ux$ \emph{and} that
the free function, free constants symbols occurring in them are among  $\ta,\tc$.
Notations like $t(\tu/\ux),\phi(\tu/\ux), \dots$ or $t(u_1/x_1, \dots, u_n/x_n),\phi(u_1/x_1, \dots, u_n/x_n), \dots$ - or occasionally just $t(\tu),\phi(\tu), \dots$ 
if confusion does not arise - are used for simultaneous substitutions within terms and \formulae.
For a given natural number $n$,  we use the standard abbreviations $\bar n$ and $n*y$ to denote the numeral of $n$ (i.e. the term $s^n(0)$, where $s$ is the successor function) and the sum of $n$ addends all equal to $y$, 
respectively.
If confusion does not arise, we may write just $n$ for $\bar n$.

By a \emph{definable function symbol}, we mean the following. Take a quantifier-free formula $\phi(\uj, y)$ such that $\forall \uj\exists ! y \phi(\uj, y)$ is valid  ($\exists !y$ stands for `there 
is a unique $y$ such that ...'). Then a definable function symbol $F$ (defined by $\phi$) is a fresh function symbol, matching the length of $\uj$ as arity, which is constrained 
to be interpreted in such a way that the formula
$
\forall y. F(\uj)= y \leftrightarrow \phi(\uj, y)
$
is true. The addition of definable function symbols does not
affect
decidability of quantifier-free \formulae and can be used for various purposes, for instance in order to express directly case-defined functions, array updates, etc. 
For instance, if $a$ is a unary free function symbol, the term 
$
wr(a,i,x)
$
(expressing the update of the array $a$ at position $i$ by over-writing $x$)
is a  definable function; formally, we have $\uj:=i,x,j$ and $\phi(\uj,y)$ is given by  $(j=i  \wedge y=x) \vee (j\neq i \wedge y=a(j))$. This formula  
$\phi(\uj, y)$ (and similar ones) can be  abbreviated like  
\vspace{-0.15cm}
$$
y~=~(\mathtt{if}~j=i~\mathtt{then}~x~\mathtt{else}~a(j))
\vspace{-0.15cm}
$$
to improve readability. Another useful definable function is integer division by a fixed natural number $n$:
to show that integer division by $n$ is definable, recall that in Presburger arithmetic we have that 
$\forall x~ \exists ! y~\bigvee_{r=0}^{n-1} (x= n*y+ 
r)$ is valid.

\section{Programs representation}\label{sec:guarded}
As a first step towards our main definability result, we provide
a classification of the relations we are interested in.
Such relations are guarded assignments required to model programs
handling arrays of unknown length.

In our framework a \emph{program} $\mathcal P$ is represented by a tuple $(\tv, l_I, l_E, T)$;
the tuple
$\tv:=\ta, \tc, pc$ models system variables; formally, we have that
\begin{description}
\item[-] the tuple $\ta=a_1, \dots, a_s$ contains free unary function symbols, i.e., the arrays manipulated by the program;
\item[-] the tuple $\tc=c_1, \dots, c_t$ contains free constants, i.e., the integer data manipulated by the program;
\item[-] the additional free constant $pc$ (called \emph{program counter}) is constrained to range over a finite set $L=\{l_1, ..., l_n\}$ of \emph{program locations} over which we distinguish the {\it initial} and {\it error} locations denoted by $l_I$ and $l_E$, respectively.
\end{description}
%
$T$ is a set of finitely many \formulae $\{\tau_1(\tv, \tv'),
\dots, \tau_r(\tv, \tv')\}$ called \emph{transition \formulae}
representing the program's body (here $\tv'$ are renamed copies of the $\tv$ representing the next-state variables).
$\mathcal P=(\tv, l_I, l_E, T)$ is \emph{safe} iff there is no satisfiable formula like
$$
(pc^0 = l_I) \wedge \tau_{i_1}(\tv^0, \tv^1)\wedge \cdots \wedge
\tau_{i_N}(\tv^{N-1}, \tv^N) \wedge (pc^N = l_E)
$$
where $\tv^0, \dots, \tv^N$ are renamed copies of the $\tv$ and each
$\tau_{i_h}$ belongs to $T$.

%
%
Sentences denoting sets of states reachable by $\mathcal{P}$ can be:
\begin{description}
\item[-] \emph{ground} sentences, i.e., sentences of the kind $\phi(\tc, \ta, pc)$;
\item[-] \emph{$\Sigma^0_1$-sentences}, i.e., sentences of the
form $\exists \ui.\; \phi(\ui, \ta, \tc, pc)$;
\item[-] \emph{$\Sigma^0_2$-sentences}, i.e., sentences of the form
$\exists \ui\, \forall \uj.\; \phi(\ui, \uj, \ta, \tc, pc)$.
\end{description}
We remark that in our context satisfiability can be fully decided only
for ground sentences and $\Sigma^0_1$-sentences (by Skolemization, as a
consequence of the general combination results \cite{NelOpp, TinHar}),
while only subclasses of $\Sigma^0_2$-sentences enjoy a decision
procedure \cite{BradleyMS06, GeM09}.
%
Transition \formulae can also be classified in three groups:
\begin{description}
\item[-] \emph{ground assignments}, i.e., transitions of the form
\begin{equation}\label{eq:transition_g}
pc=l~\wedge~ \phi_L(\tc, \ta) ~\wedge~ pc'=l' ~\wedge~
 \ta'=\lambda j.\; G(\tc, \ta, j) ~\wedge
 ~ \tc'= H(\tc, \ta)
\end{equation}
\item[-] \emph{$\Sigma^0_1$-assignments}, i.e., transitions of the form
 \begin{equation}
  \label{eq:transition_e}
 \exists \uk\;\left(
  \begin{split}
 pc=l ~\wedge~
  \phi_L(\tc,\ta, \uk) ~\wedge~ pc'=l' ~\wedge~~~~~
  \\
  \ta'=\lambda j.\; G(\tc,\ta, \uk, j) ~\wedge
  ~ \tc'= H(\tc,\ta, \uk)
 \end{split}
 \right)
  \end{equation}
 \item[-] \emph{$\Sigma^0_2$-assignments}, i.e., transitions of the form
 \begin{equation}
  \label{eq:transition_eu}
 \exists \uk\;\left( 
\begin{split}
 pc=l ~\wedge~
  \phi_L(\tc,\ta, \uk) ~\wedge~ \forall \uj~\psi_U(\tc,\ta, \uk, \uj)~\wedge~~~~~~
  \\
   pc'=l' ~\wedge~
  \ta'=\lambda j.\; G(\tc,\ta,\uk, j) 
  \wedge~ \tc'= H(\tc,\ta, \uk)
 \end{split}
 \right)
\end{equation}
\end{description}
\noindent
where $G=G_1, \dots, G_s$, $H=H_1, \dots, H_t$
are tuples of definable functions
(vectors of equations like
$\ta'=\lambda j.\; G(\tc, \ta, \uk j)$
can be replaced by the corresponding first order sentences 
$\forall j. ~\bigwedge_{h=1}^s a_h'(j)=G_h(\tc, \ta,\uk, j)$).

The \emph{composition} $\tau_1\circ \tau_2$ of two transitions $\tau_1(\tv, \tv')$ and $\tau_2(\tv, \tv')$ is 
expressed by the formula $\exists \tv_1 (\tau_1(\tv, \tv_1)\wedge \tau_2(\tv_1, \tv'))$ (notice that  composition  may result in an inconsistent 
formula, e.g., in case of location mismatch).
The \emph{preimage} $Pre(\tau, K)$ of the set of states
satisfying the formula $K(\tv)$ along the transition $\tau(\tv, \tv')$ is the set of states satisfying the formula $\exists \tv' (\tau(\tv, \tv')
\wedge K(\tv'))$. The following proposition is immediate by straightforward syntactic manipulations:

\begin{proposition}\label{prop:closure}
 Let $\tau, \tau_1, \tau_2$ be transition \formulae and let $K(\tv)$ be a formula. We have that:
{\rm (i)} if
$\tau_1, \tau_2, \tau, K$ are ground, then  $\tau_1\circ \tau_2$ is a ground assignment and $Pre(\tau, K)$ is a ground formula;
{\rm (ii)} if  $\tau_1, \tau_2, \tau, K$ are $\Sigma^0_1$,  then $\tau_1\circ \tau_2$ is a $\Sigma^0_1$-assignment and $Pre(\tau, K)$ is a $\Sigma^0_1$-sentence;
{\rm (iii)} if $\tau_1, \tau_2, \tau, K$ are  $\Sigma^0_2$, then  $\tau_1\circ \tau_2$ 
  is a $\Sigma^0_2$-assignment and $Pre(\tau, K)$ is a $\Sigma^0_2$-sentence.
\end{proposition}



\section{Backward search and acceleration}\label{sec:backward}
%
This section demonstrates the practical applicability of the
theoretical results of the paper in program verification. In particular, 
it presents the application of the accelerated transitions during
reachability analysis for guarded-assignments
representing programs handling arrays.
For readability, we first present a basic
reachability procedure. We subsequently analyze the divergence
problems and show how acceleration can be applied to solve them.
Acceleration application is not straightforward, though. The
presence of accelerated transitions might generate undesirable
$\Sigma^0_2$-sentences. The solution we propose is to over-approximate
such sentences by adopting a selective instantiation schema, known in
literature as {\it monotonic abstraction}. An enhanced reachability
procedure integrating acceleration and monotonic abstraction 
concludes the Section.

The
methodology we
exploit to check safety of a program $\mathcal P=(\tv, l_I, l_E, T)$ is \emph{backward search}: we successively
explore, through symbolic representation, all states leading to the
error location $l_E$
in one step, then in two steps,
in three steps, etc. until either we find a fixpoint or until we reach $l_I$.
To do this properly, it is
convenient to build a tree: the tree has arcs labeled by transitions and
nodes labeled by \formulae over $\tv$.
Leaves of the
tree might be marked `checked', `unchecked' or `covered'. The tree is
built according to the following non-deterministic rules.

\centerline{\textsc{Backward Search}}
\begin{description}
 \item \textsc{Initialization}: 
 a single node tree labeled by $pc=l_E$ and is marked `unchecked'. 
 \item \textsc{Check}: pick an unchecked leaf $L$ labeled with $K$.
 If $K \wedge pc=l_I$ is satisfiable (`safety test'), exit and return {\sf unsafe}. If it is not satisfiable, check whether there is
 a set $S$ of uncovered
 nodes such that (i) $L \not \in S$ and (ii) $K$ is inconsistent with the conjunction of the negations of the \formulae
 labeling the nodes in $S$ (`fixpoint check'). If it is so, mark $L$ as `covered' (by $S$). Otherwise, mark $L$ as `checked'.
  \item \textsc{Expansion}: pick a checked leaf $L$ labeled with $K$. For each transition $\tau_i \in T$, add a new leaf
  below $L$ labeled with $Pre(\tau_i, L)$ and marked as `unchecked'. The arc between $L$ and the new leaf is labeled with $\tau_i$.
  \item \textsc{Safety Exit}: if all leaves are covered, exit and return {\sf safe}.
\end{description}
The algorithm may not terminate (this is unavoidable 
by well-known undecidability results). Its correctness depends on the possibility of discharging safety tests with complete algorithms. 
By Proposition~\ref{prop:closure}, if transitions are 
ground- or $\Sigma^0_1$-assignments, completeness of safety tests
arising during the backward reachability procedure is guaranteed by the
fact that satisfiability of $\Sigma^0_1$-\formulae is decidable. 
For fixpoint tests, 
sound but incomplete algorithms may compromise termination, but not correctness of the answer; hence for fixpoint tests, we can 
adopt incomplete pragmatic algorithms (e.g. if in fixpoint tests we need to test satisfiability
of $\Sigma^0_2$-sentences, the obvious strategy is to Skolemize existentially quantified variables and to instantiate the universally quantified
ones over sets of terms chosen according to suitable heuristics). To sum up, we have:

\begin{proposition}\label{prop:correctedness}
 The above \textsc{Backward Search} procedure is partially correct for programs whose transitions are $\Sigma^0_1$-assignments, 
 i.e.,
 when the procedure terminates
 it gives a correct information about the safety of the input program.
\end{proposition}
%
Divergence phenomena  are usually not due to incomplete algorithms for fixpoint tests 
(in fact, divergence persists even in cases where fixpoint tests are precise).
%
%
\begin{example}\label{ex:reverse}
{\small
Consider a running example in \figurename~\ref{fig:motivating}(b): it reverses
the content of the array ${\tt I}$ into ${\tt O}$. In our formalism, it is
represented by the following transitions\footnote{For readability, we omit identical
updates like $I'=I$, etc. Notice that we have $l_I=1$ and $l_E=4$.}:
 $$
\begin{aligned}
& \tau_1 \equiv ~{\tt pc} = 1 \wedge {\tt pc}' = 2 \wedge {\tt c}' = 0 \\
& \tau_2 \equiv ~{\tt pc} = 2 \wedge {\tt c} \neq N+1  \wedge {\tt c}' = {\tt c} + 1 \wedge O' =  wr(O, {\tt c}, I(N-{\tt c})) \\
& \tau_3 \equiv ~{\tt pc} = 2 \wedge {\tt c} = N+1 \wedge {\tt pc}' = 3 \\
& \tau_4 
\equiv ~{\tt pc} = 3 
\wedge 
       \exists z_1\geq 0, z_2\geq 0~ (z_1+z_2=N\wedge I(z_1)\neq O(z_2)~)
             \wedge {\tt pc}' = 4.
\end{aligned}
$$
 Notice that $\tau_1-\tau_3$ all are ground assignments; only $\tau_4$ (that translates the error condition) is a $\Sigma^0_1$-assignment.
 If we apply our tree generation procedure, we get an infinite branch, whose nodes - after routine simplifications - are labeled as follows 
  $$
\begin{aligned}
& \cdots \\
(K_i)& ~~{\tt pc} = 2  \wedge \exists z_1, z_2\;\psi(z_1, z_2)\wedge {\tt c} = N-i \wedge z_2\neq N \wedge\cdots\wedge z_2\neq N-i\\
& \cdots
\end{aligned}
$$
}
where $\psi(z_1, z_2)$ stands
for $z_1 \geq 0 \wedge z_2 \geq 0 \wedge z_1+z_2=N\wedge I(z_1)\neq O(z_2)$. 
\qed
\end{example}
\COMMENT{
\begin{figure}[t]
{\scriptsize
$$
\begin{aligned}
&  ~~~~{\sf function}~{\tt Reverse}~{\tt
~int~I[N+1];~int~O[N+1];~int~c;~
} \\
& 1~~~~~~{\tt c} = 0; \\
& 2~~~~~~{\tt while}~( {\tt c} \neq N+1 )~\{O[{\tt c}] = I[N-{\tt c}];~{\tt c\!+\!+};\}  \\
& 3~~~~~~{\tt if}~(~\exists z_1\geq 0, z_2\geq 0~ (z_1+z_2=N\wedge I[z_1]\neq O[z_2]~)~) \\
& 4~~~~~~~~~{\tt ERROR;}
\end{aligned}
$$
}
\caption{Pseudo-code for function ${\tt Reverse}$. 
}
\label{fig:reverse}
\end{figure}
}
As demonstrated by the above example, a divergence source comes from the fact that we are unable to represent \emph{in one shot} the effect of 
executing   
finitely many times a given sequence of transitions.
Acceleration
can solve this problem.

\begin{definition}
The $n$-th composition of a transition $\tau(\tv, \tv')$ with itself is
recursively defined by $\tau^1:= \tau$ and $\tau^{n+1}:= \tau\circ
\tau^n$. The \emph{acceleration} $\tau^+$ of $\tau$ is $\bigvee_{n\geq
1} \tau^n$.
\end{definition}
In general, acceleration requires a logic supporting infinite disjunctions.
Notable exceptions are witnessed by Theorem~\ref{thm:tau+}. For now we
focus on examples where
accelerations yield $\Sigma^0_2$-assignments starting from ground
assignments.

\begin{example}
{\small
Recall transition $\tau_2$ from the running example.
$$
\tau_2 \equiv ~{\tt pc} = 2 \wedge {\tt c} \neq N+1 \wedge pc' =2 \wedge {\tt c}' = {\tt c} + 1 \wedge I' =I \wedge O' =  wr(O, {\tt c}, I(N-{\tt c}))
$$
(here we displayed identical updates for completeness).
Notice that the variable ${\tt pc}$ is left unchanged in this transition (this is essential, otherwise the acceleration
gives an inconsistent transition that can never fire). If we accelerate it,
we get the $\Sigma^0_2$-assignment\footnote{This $\Sigma^0_2$-assignment 
can be automatically computed using procedures outlined in the proof of
Theorem~\ref{thm:tau+}.}
\vspace{-0.2cm}
\begin{equation}\label{eq:acc_reverse}
\exists n>0~\left( 
\begin{split}
{\tt pc} = 2~ \wedge~ \forall j~({\tt c}\leq j < {\tt c}+n \to j \neq N+1)~\wedge~{\tt c}'= {\tt c}+n~\wedge ~~~~
\\
\wedge
~{\tt pc}' = 2~\wedge
~O' = \lambda j~({\tt if} ~{\tt c}\leq j < {\tt c}+ n~{\tt then}~
I(N\!-\!j)~{\tt else}~O(j))
\end{split}
\right)
\end{equation}
\vspace{-0.4cm}
\qed
}
\end{example}
In presence of
these accelerated $\Sigma^0_2$-assignments, \textsc{Backward Search} can produce
problematic $\Sigma^0_2$-sentences (see
Proposition~\ref{prop:closure} above) which cannot be handled
precisely by existing solvers.
As a solution to this problem we propose applying to such sentences a
suitable abstraction, namely {\it monotonic abstraction}.
\vspace{-0.4cm}
\begin{definition}\label{def:ma}
 Let 
$\psi:\equiv \exists \ui\, \forall \uj.\; \phi(\ui, \uj, \ta, \tc, pc)$
be a $\Sigma^0_2$-sentences and let $\mathcal S$ be a finite set   of terms of the kind $t(\ui, \tv)$. The \emph{monotonic $\mathcal S$-approximation 
of $\psi$} is the $\Sigma^0_1$-sentence
\begin{equation}\label{eq:acc}
\exists \ui\; \bigwedge_{\sigma: \uj\to \mathcal S}\phi(\ui, \uj\sigma/\uj, \ta, \tc, pc)
\end{equation}
(here $\uj\sigma$, if $\uj=j_1, \dots, j_n$, is the tuple of terms $\sigma(j_1), \dots, \sigma(j_n)$).
\end{definition}
By Definition \ref{def:ma}, 
universally quantified variables are \emph{eliminated through instantiation};
the larger the set $\mathcal S$ is, the better approximation you get. In
practice, the natural choices for $\mathcal S$ are $\ui$ or the set of terms of the kind $t(\ui, \tv)$ occurring in $\psi$
(we adopted the former choice in our implementation).
As a result of replacing $\Sigma^0_2$-sentences by their monotonic
approximation, spurious unsafe traces might occur.
However, 
\textbf{those can be disregarded if accelerated transitions contribute
to their generation}.
This is because if $\mathcal P$ is unsafe, then unsafety can be discovered without appealing to accelerated transitions.

To integrate monotonic abstraction, the above \textsc{Backward Search} procedure is modified as follows.
In a \textsc{Preprocessing} step, we add
some accelerated transitions of the kind
$(\tau_1\circ \cdots \circ \tau_n)^+$ to $T$.
These transitions can be found by
inspecting cycles in the control flow graph of the program and
accelerating them following the procedure described in Sections 5, 6.
The natural cycles to inspect are those corresponding to loop branches in the source code. It should be noticed, however,
that identifying the good cycles to accelerate is subject to specific heuristics that deserve separate investigation in
case the program has infinitely many cycles. (choosing cycles from branches of innermost loops 
is the simplest example of such heuristics and the one
we implemented).

After this extra preprocessing step, the remaining instructions are left unchanged, with the exception of \textsc{Check} that is modified as follows:
\begin{description}
 \item \textsc{Check}': pick an unchecked leaf $L$
labeled by a formula $K$. If $K$ is a $\Sigma^0_2$-sentence,
choose a suitable $\mathcal S$ and 
  replace $K$ by its monotonic $\mathcal S$-abstraction $K'$.
 If $K'\wedge pc=l_I$ is inconsistent, mark $L$ as `covered' or `checked' according to the outcome of the fixpoint check, as was done in the original {\sc Check}.
 If $K'\wedge pc=l_I$ is satisfiable, analyze the path from the root to $L$. If 
 no accelerated transition $\tau^+$ is found in it  return {\sf unsafe}, otherwise
%
remove the sub-tree $D$ from the target of 
$\tau^+$ 
to the leaves. Each node $N$ covered by a node in $D$ will be flagged as
 `unchecked'
 (to make it eligible in future for the \textsc{Expansion} instruction).
 
\end{description}
The new procedure will be referred as {\sc Backward Search}'.
It is quite straightforward to see that
Proposition~\ref{prop:correctedness} still applies to the modified 
algorithm.
Notice that, although termination cannot be ensured (given well-known undecidability results), spurious traces containing approximated 
accelerated transitions cannot be produced again and again: when the sub-tree $D$ from the target node $v$ of 
$\tau^+$ is removed by \textsc{Check}', the node $v$ is not a leaf (the arcs labeled by the transitions $\tau$ are still there),  hence it cannot be expanded
anymore according to the \textsc{Expansion} instruction.

\begin{example}
{\small
Let again consider our running example and demonstrate how acceleration
and monotonic abstraction work. In the preprocessing step, we add the accelerated transition $\tau_2^+$ given by~\eqref{eq:acc_reverse}
to the transitions
we already have. After having  computed $(K')\equiv Pre(\tau_4, K), (K'')\equiv Pre(\tau_3, K')$, 
 we compute $(\tilde K)\equiv Pre(\tau_2^+, K'')$ and get
\vspace{-0.05cm}
\begin{equation*}
\exists n>0\,\exists z_1, z_2\,\left( 
\begin{split}
{\tt pc} = 2~ \wedge~ \forall j~({\tt c}\leq j < {\tt c}\!+\!n \to j \neq N\!+\!1)~\wedge~~~~~~~~
\\
\wedge~ {\tt c}\!+\!n= N\!+\!1~\wedge z_1 \geq 0 ~\wedge~ z_2\geq 0 ~\wedge z_1+z_2=N~\wedge~
\\
\wedge
~I(z_1) \neq \lambda j~({\tt if} ~{\tt c}\leq j < {\tt c}+ n~{\tt then}~
I(N\!-\!j)~{\tt else}~O(j))(z_2)
\end{split}
\right)
\end{equation*}
We approximate using the set of terms $\mathcal S=\{z_1, z_2, n\}$. 
After
simplifications we get
\vspace{-0.1cm}
\begin{equation*}
 \exists z_1, z_2\, ({\tt pc} = 2~\wedge ~ {\tt c}\leq N
~\wedge z_1 \geq 0 ~\wedge~ z_2\geq 0 ~\wedge~
 z_1+z_2=N~\wedge~O(z_2)\neq I(z_1)~\wedge 
~{\tt c}>z_2)
\vspace{-0.1cm}
\end{equation*}
Generating this formula is enough to stop divergence.
\qed
}
\end{example}

Notice that in the computations of the above example we eventually succeeded in eliminating the extra quantifier
$\exists n$ introduced by the accelerated transition. This is not always possible:
sometimes in fact, to get the good invariant one needs more quantified variables than those occurring 
in the annotated program and accelerated transitions might be the way of getting 
such additional quantified variables.
As an example of this phenomenon, consider 
the {\tt init+test} program
included in our benchmark suite of Section~\ref{sec:experiments} below.

\section{Iterators}\label{sec:iterators}
This Section introduces \emph{iterators} and \emph{selectors}, two
main ingredients used to
supply a useful format to compute accelerated transitions.
Iterators are meant to formalize the notion of a counter scanning the indexes of an array: the most simple iterators are increments and decrements, but one may also build
more complex ones for different scans, like in binary search. We give
their formal definition and then we
supply
some examples. We need to handle tuples of terms because we want to consider the case 
in which we deal with different arrays with possibly different scanning variables.
Given a $m$-tuple of terms 
\begin{equation}\label{eq:terms}
\tu(\ux)~:=~u_1(x_1, \dots, x_m), \dots, u_m(x_1, \dots, x_m)
\end{equation}
containing the $m$ variables $\ux=x_1, \dots, x_m$,
we indicate with $\tu^n$ the term expressing the $n$-times composition of (the function denoted by) $\tu$ with itself.
Formally, we have $\tu^0(\ux):=\ux$ and
$$
\tu^{n+1}(\ux)~:=~u_1(\tu^n(\ux)), \dots, u_m(\tu^n(\ux)) ~.
$$

\begin{definition}
A tuple of terms $\tu$ like~\eqref{eq:terms}
is said to be an  \emph{iterator} iff there exists an $m$-tuple of $m+1$-ary terms
$
\tu^*(\ux,y)~:=~u^*_1(x_1, \dots, x_m,y), \dots, u^*_m(x_1, \dots, x_m,y)
$
such that for
any
natural number $n$ it happens that the formula
\begin{equation}\label{eq:iterators}
\tu^n(\ux)= \tu^*(\ux, \bar n)
\end{equation}
is valid.\footnote{Recall 
that $\bar n$ is the numeral of $n$, i.e. it is $s^n(0)$.}
Given an iterator $\tu$ as above, we say that
an $m$-ary term $\kappa(x_1, \dots, x_m)$
is a \emph{selector} for $\tu$ iff
there is an $m+1$-ary term $\iota(x_1, \dots, x_m, y)$ 
yielding the validity of 
the formula
\begin{equation}\label{eq:iter}
 z= \kappa(\tu^*(\ux,y))
 \to 
 y=\iota(\ux, z)
~~.
 \end{equation}
\end{definition}

The meaning of condition~\eqref{eq:iter} is that, once the input $\ux$
and the selected output $z$ are known, it is possible to identify
uniquely (through $\iota$) the number of iterations $y$ that are needed to get $z$ by applying $\kappa$ to $\tu^*(\ux,y)$. 
The term $\kappa$ is a selector function that selects (and possibly
modifies) one of the $\tu$;
in most applications (though not always) $\kappa$ is a projection, represented as a variable
$x_i$ (for $1\leq i \leq m$), so that $\kappa(\tu^*(\ux,y))$ is just the $i$-th component $u^*_i(\ux,y)$ of the tuple of terms $\tu^*(\ux, y)$.
In these cases, the formula~\eqref{eq:iter}
reads as 
\begin{equation}\label{eq:iter1}
z= u^*_i(\ux,y) 
\to 
y=\iota(\ux, z)~.
\end{equation}

\begin{example}\label{ex:standard}
{\small
The canonical example is when we have 
$m=1$ and $\tu:=u_1(x_1):=x_1+1$; this is an iterator with $u_1^*(x_1,y):=x_1+y$; as a selector, we can take 
$\kappa(x_1):=x_1$ and $\iota(x_1,z):=z-x_1$. 
}
\qed
\end{example}

\begin{example}\label{ex:modulo}
{\small
The previous example can be modified, by choosing  $\tu$ to be $x_1+\bar n$, for some 
integer $n\neq 0$: then we have $u^*(x_1,y):= x_1+ n*y$, $\kappa(x_1):= x_1$, and $\iota(x_1,z)= (z-x_1)//n$
where // is integer division (recall that integer division by a given $n$
is definable in Presburger arithmetic).
}
\qed
\end{example}

\begin{example}
{\small
If we move to more expressive arithmetic theories, like Primitive Recursive Arithmetic (where we have a symbol for every
primitive recursive function), we can get much more  examples.
As an example with $m>1$, we can take 
$\tu:=x_1+x_2, x_2$ and get $u_1^*(x_1,x_2,y)=x_1+y*x_2$, $u_2^*(x_1,x_2,y)=x_2$. Here a selector is for instance $\kappa_1(x_1, x_2):=\bar 7+x_1$,
$\iota(x_1, x_2,z):= (z-x_1-\bar 7)//x_2$.
}
\qed
\end{example}

\section{Accelerating local ground assignments}\label{sec:accelerating}
Back to our program $\mathcal P=(\tv, l_I, l_E, T)$,
we look for 
conditions on transitions from $T$ allowing to accelerate them 
via a $\Sigma^0_2$-assignment. 
Given an iterator $\tu(\ux)$, a \emph{selector assignment} for $\ta:=a_1, \dots, a_s$ (relative to $\tu$) is a tuple
of selectors $\kappa:=\kappa_{1}, \dots, \kappa_{s}$ for $\tu$.
Intuitively,
the components  of the tuple are meant to 
indicate the scanners of the arrays $\ta$
and 
as such 
might not be distinct (although, of course, just \emph{one} selector is assigned to each array).
A formula $\psi$ (resp. a term $t$) is said to be \emph{purely arithmetical} over a finite set of terms $V$ 
iff it is obtained from a formula (resp. a term)  \emph{not containing the extra free function symbols $\ta, \tc$}
by replacing some free variables in it by terms from $V$.
Let $\tv=v_1, \dots, v_s$ and $\tw=w_1, \dots, w_s$ be $s$-tuples of terms; below 
$wr(\ta, \tv, \tw)$ and $\ta(\tv)$ indicate the tuples $wr(a_1, v_1, w_1), \dots, wr(a_s, v_s, w_s)$ and $a_1(v_1),\dots, a_s(v_s)$,
respectively (recall  from Section~\ref{sec:guarded} that $s=\vert \ta\vert$).

\begin{definition}\label{def:localground}
 A \emph{local ground assignment} is a ground assignment of the form 
 \begin{equation}\label{eq:transition_sg}
 pc=l~\wedge~ \phi_L(\tc, \ta) ~\wedge~ pc'=l ~\wedge~
  \ta'=wr(\ta,\kappa(\tilde\tc), \ttb(\tc, \ta)) ~\wedge
  ~ \tilde\tc'= \tu(\tilde\tc) ~\wedge \td'=\td
 \end{equation}
 where 
{\rm (i)} $\tc=\tilde \tc, \td$;
{\rm (ii)}
$\tu 
=u_1, \dots,u_{\vert \tilde \tc\vert}$ is an iterator; 
{\rm (iii)} 
 the terms $\kappa$ are a selector assignment 
  for $\ta$ relative to $\tu$;
{\rm (iv)}
the formula $\phi_L(\tc, \ta)$ and the terms  $\ttb(\tc, \ta)$ are purely arithmetical over the set of terms $\{\tc, \ta(\kappa(\tilde\tc))\}
\cup \{a_i(d_j)\}_{1\leq i\leq s, 1\leq j\leq \vert \td\vert}$; {\rm (v)} the guard $\phi_L$ contains the conjuncts
$\kappa_i(\tilde \tc)\neq d_j$, for $1\leq i\leq 
s$
and $1\leq j\leq \vert \td\vert$. 
\end{definition}

Thus in a local ground assignment, there are various restrictions: (a) the numerical variables are split into `idle' variables $\td$
and variables $\tilde \tc$ subject to update via an iterator $\tu$; 
(b) the program counter is not modified; (c) the guard does not depend on the values of the $a_i$ at 
cells different from   $\kappa_i(\tilde\tc), \td$; (d) the update of the $\ta$ 
are simultaneous writing operations modifying only the  entries $\kappa(\tilde\tc)$.
Thus, the assignment is local
and the relevant modifications it makes are determined by the selectors locations.
 The
`idle' variables $\td$ are useful to accelerate branches of nested loops; 
the inequalities mentioned in (v) are automatically generated by making case distinctions in assignment guards.

\begin{example}
{\small
For our running example,
we show that transition $\tau_2$ (the one we want to accelerate) is a local ground assignment.
We have $\td=\emptyset$ and $\tilde\tc={\mathtt c}$ and $ \ta=I,O$.
The 
 counter ${\tt c}$  is incremented by 1 at each application of $\tau_2$. Thus, our iterator is $\tu:= x_1+1$ and the selector assignment assigns $\kappa_1:=N-x_1$ to $I$ and $\kappa_2:=x_1$ to $O$. In this way, $I$ is modified
 (identically) at $N-{\tt c}$ via $I'=wr(I,N-{\tt c}, I(N-{\tt c}))$ and $O$ is modified at ${\tt c}$ via $O'=wr(O, {\tt c}, I(N-{\tt c }))$. The guard $\tau_2$ is  ${\tt c}\neq  N+1$.
Since the formula ${\tt c}\neq  N+1$ and the term $I(N-{\tt c })$ are purely arithmetical over  $\lbrace {\tt c}, I(N-{\tt c }), O({\tt c})\rbrace$, we conclude that  $\tau_2$ is local. 
}
\qed
\end{example}

\begin{theorem}\label{thm:tau+}
If $\tau$ is a local ground assignment, then $\tau^+$ is a
$\Sigma^0_2$-assignment.
\end{theorem}

\begin{proof} (Sketch, see Appendix~\ref{app:proofs} for full details).
Let us fix the local ground assignment~\eqref{eq:transition_sg}; let $\ta[\td]$ indicate the $s*\vert\td\vert$-tuple of terms 
$\{a_i(d_j)\}_{1\leq i\leq s, 1\leq j\leq \vert \td\vert}$;
since $\phi_L$ and $\ttb
:= t_1, \dots, t_s$ are purely arithmetical over  $\{\tilde\tc,\td, \ta(\kappa(\tilde\tc)),\ta[\td] \}$, we have that they can be written as 
$\tilde\phi_L(\tilde\tc,\td, \ta(\kappa(\tilde\tc)),\ta[\td])$, $\tilde \ttb(\tilde\tc,\td ,\ta(\kappa(\tilde\tc)),\ta[\td])$, 
respectively, where $\tilde \phi_L, \tilde \ttb$ do not contain 
occurrences of the free function and constant symbols $\ta, \tc$.
The transition $\tau^+$ can be expressed as a $\Sigma^0_2$-assignment by
\begin{equation*}
  \exists y>0
  \left(
   \begin{split}
  \forall z~ (0\leq\! z\! < y \!\to\! \tilde\phi_L(\tu^*(\tilde\tc,z),\td, \ta(\kappa(\tu^*(\tc,z))), \ta[\td])
  \wedge  \td'=\td\,\wedge
  \\
  \wedge~pc=l ~\wedge~pc'=l~
  \wedge~ \tilde\tc'=\tu^*(\tilde\tc, y)~\wedge~\ta'= \lambda j.\; F(\tc, \ta,y, j)
  \end{split}
  \right)
\end{equation*}
where the tuple $F=F_1, \dots, F_s$ of definable functions is given by
\begin{equation*}
\begin{split}
 F_h(\tc, \ta,y, j)~=~\lambda j.~~ \mathtt{if}~~ 0 \leq\iota_{h}(\tilde\tc,j) < y~ 
 \wedge j= \kappa_h(\tu^*(\tc,\iota_{h}(\tilde\tc,j)))
 ~\mathtt{then}\;
 \\   
\tilde t_{h}(\tu^*(\tilde\tc, \iota_{h}(\tilde\tc,j) ),\td, \ta(\kappa(\tu^*(\tilde\tc,\iota_{h}(\tilde\tc,j)))),\ta[\td])\;\mathtt{else}\;a_h[j]
 \end{split}
\end{equation*}
for $h=1,\dots, s$ (here $\iota_1, \dots, \iota_s$ are the terms corresponding to $\kappa_1, \dots,
\kappa_s$ according to the definition of a selector for the iterator $\tu$).
$\hfill\dashv$
\end{proof}

We point out that the effective use of Theorem~\ref{thm:tau+}  relies on the implementation of a repository of 
iterators and selectors and of
algorithms recognizing them. The larger the repository is, the more possibilities the model checker has to exploit the full power of acceleration.

In most applications it is sufficient to 
consider  accelerated transitions of the canonical form of Example~\ref{ex:standard}.
Let us examine in details this special case; here
 $\tc$ is a single counter ${\tt c}$ that is incremented by one (otherwise said, the iterator is $x_1+1$) and
 the selector assignment is trivial, namely it is just $x_1$. We call
these local ground assignments \emph{simple}.
Thus, a simple local ground assignment has the form 
\begin{equation}\label{eq:simplelocal}
 pc=l~\wedge~ \phi_L({\tt c}, \ta) ~\wedge~ pc'=l \wedge {\tt c}'={\tt c}+1~\wedge~
  \ta'=wr(\ta, {\tt c}, \ttb({\tt c},\ta)) 
 \end{equation}
where 
the first occurrence of ${\tt c}$  in $wr(\ta, {\tt c}, \ttb({\tt c},\ta))$ stands in fact for
an  $s$-tuple of terms all identical to ${\tt c}$, and where $\phi_L, \ttb$  are purely arithmetical over the terms ${\tt c}$, $a_1[{\tt c}], \dots, a_s[{\tt c}]$.
The accelerated transition  computed in the proof of Theorem~\ref{thm:tau+} for~\eqref{eq:simplelocal}
can be rewritten as follows:
\begin{equation}
  \label{eq:transition_+}
 \exists k\left( 
\begin{split}
k> 0~\wedge~
pc=l ~\wedge~
  \forall j~({\tt c}\leq j<{\tt c}+k \to\phi_L(j, \ta))~\wedge~ pc'=l ~\wedge~ 
  \\
\wedge~ {\tt c}'={\tt c}+k~\wedge~
  \ta'=\lambda j.\;(\mathtt{if}~{\tt c}\leq j<{\tt c}+k~\mathtt{then} ~\ttb(j,\ta)~\mathtt{else}~\ta[j]) 
 \end{split}
 \right)
\end{equation}
A slight extension of the notion of a simple assignment leads to
a 
further
subclass of
local ground assignments useful to accelerated branches of nested loops (see Appendix~\ref{sec:all_diff_appendix} for more details).

\section{Experimental evaluation}\label{sec:experiments}
We implemented the algorithm described in Section
\ref{sec:backward} - \ref{sec:accelerating} as a preprocessing module inside the {\sc mcmt}
model checker \cite{GhilardiR10b}.
To perform a feasibility study, we intentionally focused our implementation on 
simple and simple+ local ground assignments.
For a thorough and unbiased evaluation we compared/combined the new technique
with an abstraction algorithm suited for array programs \cite{ABGRS12a} implemented in the same
tool. This section describes benchmarks
and discusses experimental results. A clear outcome from our experiments
is that abstraction/refinement and acceleration techniques can be
gainfully combined.

\noindent{\bf Benchmarks.}~
We evaluated the new algorithm on 55 programs with
arrays, each annotated with an assertion. We considered only
quantifier-free or $\forall$-as\-ser\-tions.
Our set of benchmarks comprises programs used to evaluate
the Lazy Abstraction with Interpolation for Arrays framework
\cite{ABGRS12b}
and other focused benchmarks
where abstraction diverges. These are problems involving array
manipulations as copying, comparing, searching, sorting, initializing,
testing, etc. About one third of the programs contain bugs.\footnote{The set of benchmarks can be downloaded from \url{http://www.inf.usi.ch/phd/alberti/prj/acc}; 
the tool set \textsc{mcmt} is available at
\url{http://users.mat.unimi.it/users/ghilardi/mcmt/}.}

\noindent{\bf Evaluation.}~
Experiments have been run on a machine equipped with a i7@2.66 GHz CPU
and 4GB of RAM running OS X. Time limit for each experiment has been set to 60 seconds.
%
We run {\sc mcmt} with four different configurations:
\begin{itemize}
\item {\sc Backward Search} - {\sc mcmt} executes the procedure described
at the beginning of Section \ref{sec:backward}.
\item {\sc Abstraction} - {\sc mcmt} integrates the backward
reachability algorithm with the abstraction/refinement loop
\cite{ABGRS12a}.
\item {\sc Acceleration} - The transition system is pre-processed in
order to compute accelerated transitions (when it is possible) and then
the {\sc Backward Search}' procedure is executed.
\item {\sc Accel. + Abstr.} - This configuration enables both the
preprocessing step in charge of computing accelerated transitions and
the abstraction/refinement engine on the top of the {\sc Backward
Search}' procedure.
\end{itemize}
%
The complete statistics can be found in 
Appendix~\ref{sec:experiments_appendix}.
%
\begin{figure}[t]
\centering
\subfigure[\label{fig:bs_acc}]{\includegraphics[width=.4\linewidth]{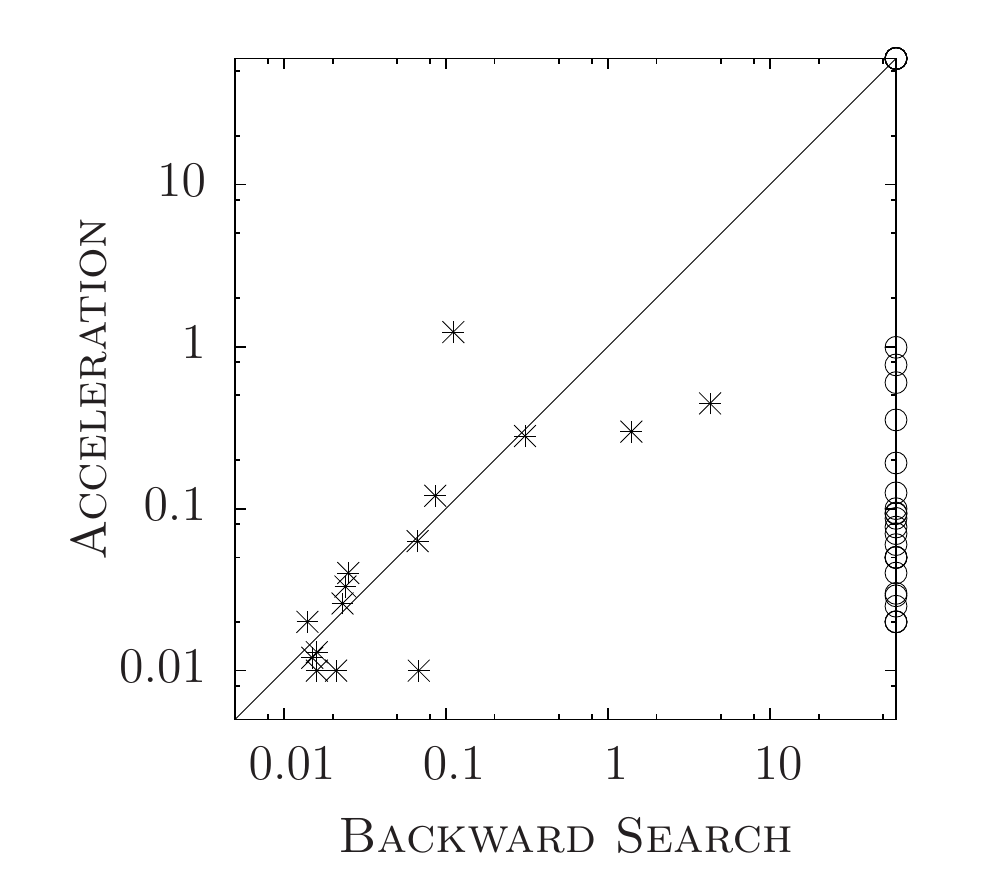}}
\subfigure[\label{fig:abs_acc}]{\includegraphics[width=.4\linewidth]{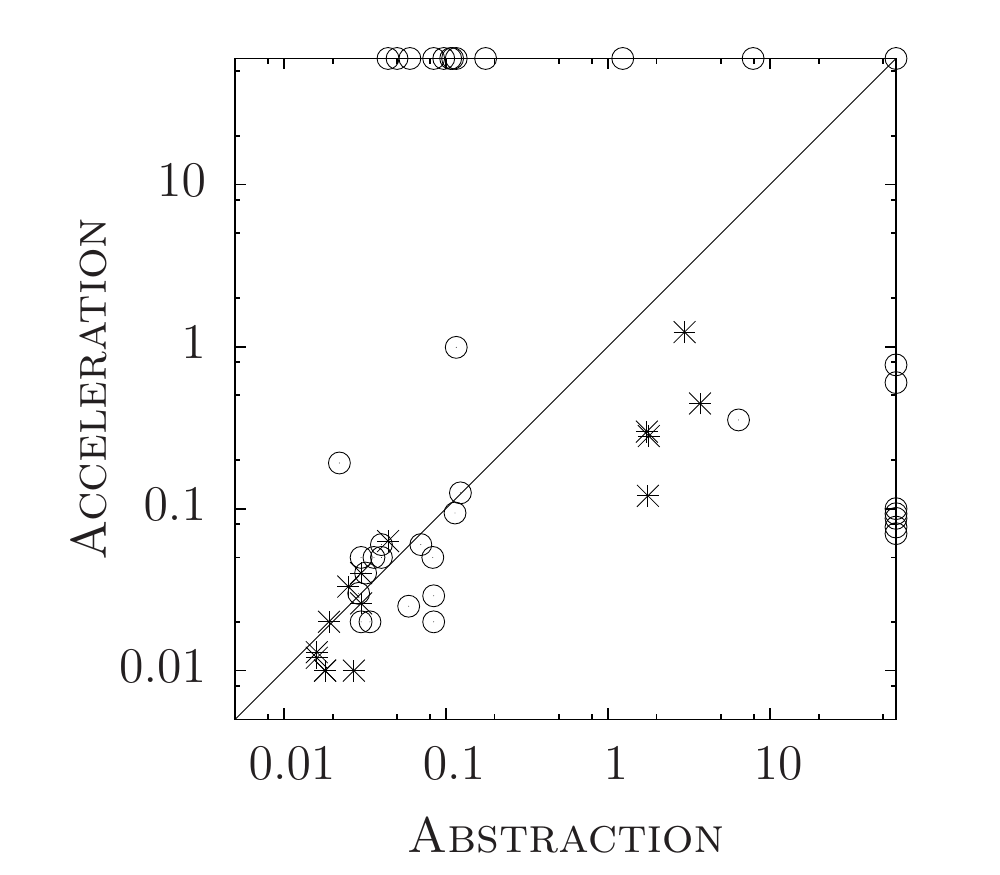}}

\subfigure[\label{fig:acc_absacc}]{\includegraphics[width=.4\linewidth]{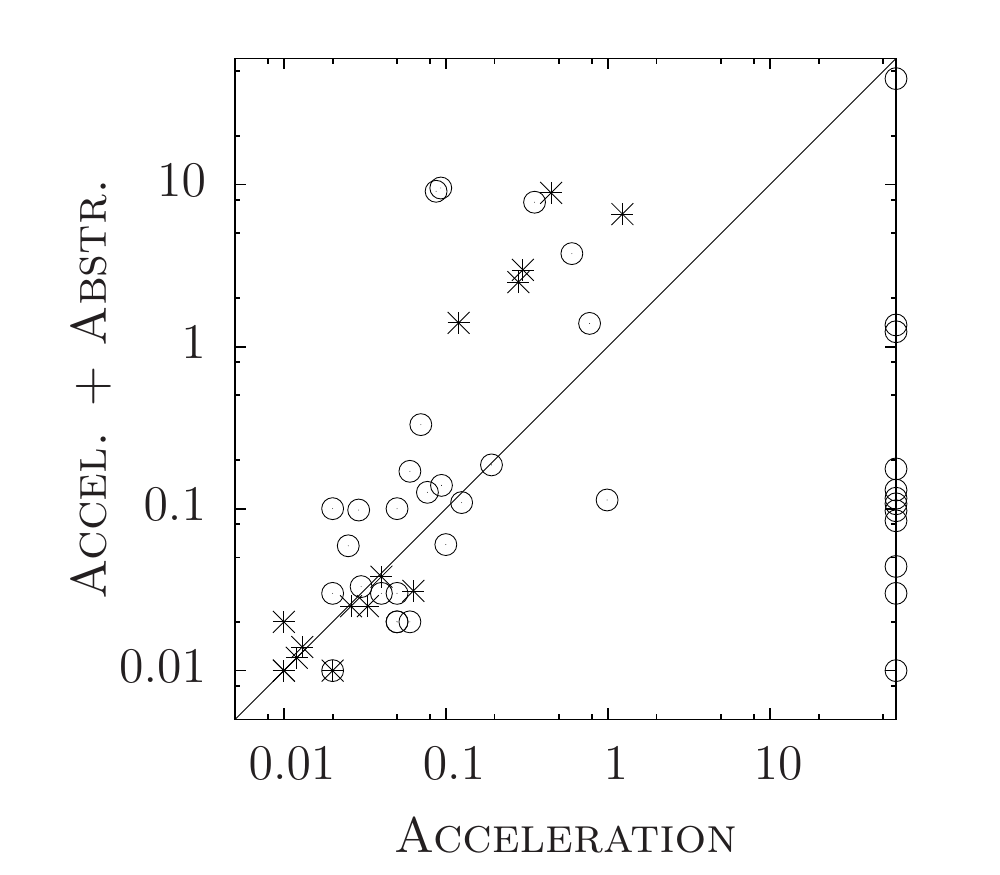}}
\subfigure[\label{fig:abs_absacc}]{\includegraphics[width=.4\linewidth]{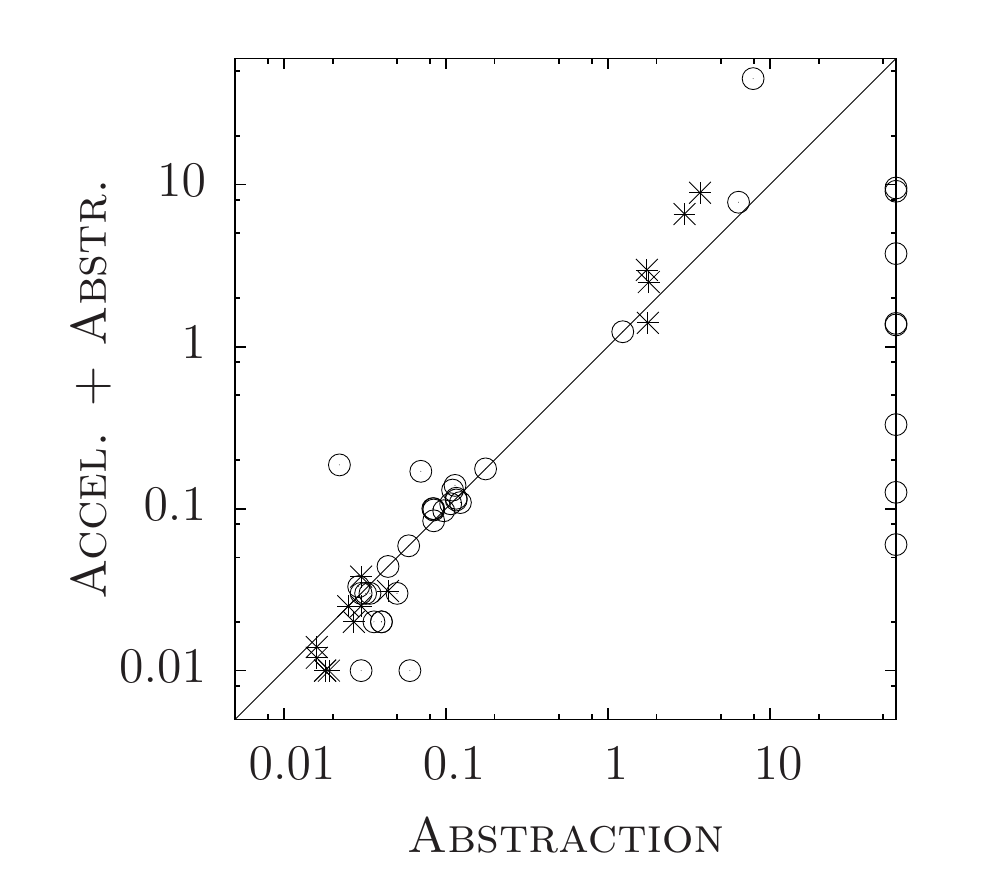}}
\caption{\label{fig:exp}
{\small 
Comparison of time for different options of {\sc Backward Search}.
Stars and circles represent buggy and correct programs
respectively.
}
}
\vspace{-0.5cm}
\end{figure}
In summary, the comparative analysis of timings presented in \figurename
\ref{fig:exp} confirms that acceleration indeed helps to avoid
divergence for problematic programs where abstraction fails.
The first comparison (\figurename \ref{fig:bs_acc}) highlights the
benefits of using acceleration: {\sc Backward Search} diverges on all 39 safe instances. Acceleration
stops divergence in 23 cases, and moreover the overhead introduced by the
preprocessing step does not affect unsafe instances.
%
\figurename \ref{fig:abs_acc}
shows that acceleration
and abstraction are two complementary techniques, since {\sc mcmt} times
out in both cases but for two different sets of programs.
%
\figurename \ref{fig:acc_absacc} and
\figurename \ref{fig:abs_absacc} attest that acceleration and
abstraction/refinement techniques mutually benefit from each other:
with both techniques
{\sc mcmt} solves all the 55
benchmarks.

\section{Conclusion and Future Work}\label{sec:conclusion}
%
We identified a class of transition relations involving
array updates that can be accelerated, showed how it is possible to
compute the accelerated transition and describe a solution for dealing
with universal quantifiers arising from the acceleration process.
%
Our paper lays theoretical
foundations for this interesting research topic and confirms by our
prototype experiments on challenging benchmarks its advantages over
stand-alone verification approaches since it's able to solve problems on
which other techniques fail to converge.

As future directions, a challenging task is to enlarge the definability result of Theorem~\ref{thm:tau+}
so as to cover classes of transitions modeling more and more loop branches arising from concrete programs. 
In addition, one may want to consider
more sophisticated strategies for instantiation in order to support acceleration.
Considering increasing larger
$\mathcal S$
or 
handling
$\Sigma^0_2$-sentences when they belong to decidable fragments
\cite{BradleyMS06,GeM09}
may lead to further improvements.

\bibliographystyle{plain}
\bibliography{references}

\newpage

\appendix

\section{Proof of Theorem~\ref{thm:tau+}}\label{app:proofs}

In this technical Appendix, we supply the  proof of Theorem~\ref{thm:tau+}.

\begin{proof}

As a preliminary observation, we notice that the bi-implications of the kind
\begin{equation}\label{eq:p1}
(\bigvee_{n\geq 0} \psi(\ux, \bar n)) \leftrightarrow \exists y\;(y\geq 0 \wedge \psi(\ux,y))~.
\end{equation}
are valid because we  interpret our \formulae in the \emph{standard} structure of natural numbers (enriched 
with extra free symbols).

As a second preliminary observation, we notice that~\eqref{eq:iter}
can be equivalently re-writtem in the form of a bi-implication as:
\begin{equation}\label{eq:iter1}
 z= \kappa(\tu^*(\ux,y))\quad
 \leftrightarrow \quad
[\;y=\iota(\ux, z) ~\wedge ~ z= \kappa(\tu^*(\ux,\iota(\ux, z)))\;]
~~
 \end{equation}
(to see why~\eqref{eq:iter1} is equivalent to~\eqref{eq:iter} it is sufficient to apply the logical laws of pure identity).

Let us fix a local ground assignment of the form~\eqref{eq:transition_sg}; let $\ta[\td]$ indicate the $s*\vert\td\vert$-tuple of terms 
$\{a_i(d_j)\}_{1\leq i\leq s, 1\leq j\leq \vert \td\vert}$;
since $\phi_L$ and $\ttb$ are purely arithmetical over  $\{\tilde\tc,\td, \ta(\kappa(\tilde\tc)),\ta[\td] \}$, we have that they can be written as 
$\tilde\phi_L(\tilde\tc,\td, \ta(\kappa(\tilde\tc)),\ta[\td])$, $\tilde \ttb(\tilde\tc,\td ,\ta(\kappa(\tilde\tc)),\ta[\td])$, respectively, where $\tilde \phi_L, \tilde \ttb$ do not contain 
occurrences of the free function and constant symbols $\ta, \tc$.

\emph{Claim.}
As a first step, 
we  show by induction on $n$ that $\tau^n$ can be expressed as follows
(we omit here and below the conjuncts $pc=l\wedge pc'=l\wedge \td'=\td$ that do not play any role)
\begin{equation}\label{eq:transition1}
\bigwedge_{0\leq k< n}
 \tilde\phi_L(\tu^*(\tilde\tc,\bar k),\td, \ta(\kappa(\tu^*(\tilde\tc,\bar k))),\ta[\td]) ~\wedge~ \tilde\tc'=\tu^*(\tilde\tc, \bar n)~\wedge~\ta'= \lambda j.\;F(\tc, \ta,\bar n, j)
\end{equation}
where the tuple $F=F_1, \dots, F_s$ of definable functions is given by\footnote{ 
The following is an informal explanation of the formula~\eqref{eq:transition_upd} expressing iterated updates.
The point is to recognize whether a given cell $j$ has been over-written or not within the first $y$ iterations. The number 
$\iota_{h}(\tilde\tc,j)$ gives the candidate number of iterations needed to get $j$ and the further condition 
$j= \kappa_h(\tu^*(\tc,\iota_{h}(\tilde\tc,j)))$ checks whether this number is correct or not. Take for instance Example~\ref{ex:modulo} 
with $n=2$. Then if we have a single counter initialized to say 4, our iterations give values $4+2, 4+2+2, \dots$ for the updated counter. 
If we want to know whether $j$ can be
reached within less than 5 iterations, we just compute $\iota(4, j)$ which is the quotient of the integer division of $j-4$ by 2. The we need to check that 
$\iota(4,j)$ is among $0, \dots, 4=5-1$ and \emph{also} that $j$ can be really reached from $\tilde \tc=4$ by adding 2 to it  $\iota(4,j)$-times 
(the latter won't be true if  $j$ is odd).
}
\begin{equation}\label{eq:transition_upd}
\begin{split}
 F_h(\tc, \ta,y, j)~=~\lambda j.~~ \mathtt{if}~~ 0 \leq\iota_{h}(\tilde\tc,j) < y~ 
 \wedge j= \kappa_h(\tu^*(\tc,\iota_{h}(\tilde\tc,j)))
 ~\mathtt{then}\;
 \\   
\tilde t_{h}(\tu^*(\tilde\tc, \iota_{h}(\tilde\tc,j) ),\td, \ta(\kappa(\tu^*(\tilde\tc,\iota_{h}(\tilde\tc,j)))),\ta[\td])\;\mathtt{else}\;a_h[j]
 \end{split}
\end{equation}
for $h=1,\dots, s$ (here $\iota_1, \dots, \iota_s$ are the terms corresponding to $\kappa_1, \dots, \kappa_s$ according to the definition of a selector for the iterator $\tu$).

\emph{ Proof of the Claim}.
For $n=1$, notice that $\tilde\phi_L(\tu^*(\tilde\tc,0),\td, \ta(\kappa(\tu^*(\tilde\tc,0))), \ta[\td])$
is equivalent to $\tilde\phi_L(\tilde\tc,\td, \ta(\kappa(\tilde\tc)),\ta[\td])$, that $\tilde\tc'=\tu^*(\tilde\tc, \bar 1)$ is equivalent to
$\tilde\tc'=\tu(\tilde\tc)$ and that $\lambda j.\;F(\tilde\tc, \td,\ta,\bar 1, j)=wr(\ta,\kappa(\tilde\tc), \ttb(\tilde\tc,\td, \ta(\kappa(\tilde\tc)),\ta[\td]))$ holds  (the latter because 
for every $h$, $\iota_{h}(\tilde\tc,j)=0 \wedge j= \kappa_h(\tu^*(\tc,\iota_{h}(\tilde\tc,j))$
is equivalent to $j=\kappa_{h}(\tu^*(\tilde\tc, 0))=\kappa_{h}(\tilde\tc)$ by~\eqref{eq:iter1}).

For the induction step, we suppose the Claim holds for $n$ and show it for $n+1$. 
As a preliminary remark, notice that from~\eqref{eq:transition_sg}, we get not only $\td'=\td$, but also $\ta'[\td']=\ta[\td]$, because of \rm{(v)} of Definition~\ref{def:localground}. As a consequence,
after $n$ iterations of $\tau$, the values $\td, \ta[\td]$ are left unchanged; thus, for notation simplicity, we will not display anymore 
below the dependence of $\phi_L, \tilde \ttb$ on $\td, \ta[\td]$. 
We need to show that $\tau\circ \tau^n$ matches the required shape~\eqref{eq:transition1}-\eqref{eq:transition_upd} 
with $n+1$ instead of $n$. After unraveling the definitions, this splits into three sub-claims,  concerning the update of the $\tc$, the guard and the update of the $\ta$, respectively:
\begin{description}
 \item[{\rm (i)}] the equality $\tu(\tu^*(\tilde\tc, \bar n)) = \tu^*(\tilde\tc, \overline{n+1})$ is valid;
 \item[{\rm (ii)}]
 $$
 \bigwedge_{0\leq k< n} \tilde\phi_L(\tu^*(\tilde\tc,\bar k), \ta(\kappa(\tu^*(\tilde\tc,\bar k)))) ~\wedge~\tilde\phi_L(\tu^*(\tilde\tc,\overline{n}), \lambda j.\;F(\tc, \ta,\bar n, j)(\kappa(\tu^*(\tilde\tc,\overline{n}))))$$
is equivalent to $$\bigwedge_{0\leq k< n+1}
 \tilde\phi_L(\tu^*(\tc,\bar k), \ta(\kappa(\tu^*(\tc,\bar k))));$$
\item[{\rm (iii)}] $wr(\lambda j.\,F(\tc, \ta,\bar n, j),\kappa(\tu^*(\tilde\tc, \overline{n})), \tilde \ttb(\tu^*(\tilde\tc, \overline{n}),\lambda j.\, F(\tc, \ta,\bar n, j)(\kappa(\tu^*(\tilde\tc, \bar n))))$ is the same
function as $\lambda j.\;F(\tc, \ta,\overline{n+1}, j)$.
\end{description}

Indeed statement (i) is trivial, because $\tu(\tu^*(\tilde\tc, \bar n))=\tu(\tu^n(\tilde\tc))=\tu^{n+1}(\tilde\tc)=\tu^*(\tilde\tc, \overline{n+1})$ holds by~\eqref{eq:iterators}.

To show (ii), it is sufficient to check that 
\begin{equation}\label{eq:eq}
\ta(\kappa(\tu^*(\tilde\tc,\overline{n})))~=~
\lambda j.\;F(\tc, \ta,\bar n, j)(\kappa(\tu^*(\tilde\tc,\overline{n})))
\end{equation}
is true. In turn, this follows from~\eqref{eq:transition_upd} and the validity of the following implications (varying $h=1,\dots,s$)
\begin{equation}\label{eq:impl}
 \iota_{h}(\tilde\tc,j) \neq \bar n ~\to~ j\neq \kappa_h(\tu^*(\tilde\tc,\overline{n}))~
\end{equation}
(in fact, $a_h$ and $F_h$ can possibly differ only for the $j$ satisfying $0 \leq\iota_{h}(\tilde\tc,j) < \bar n$,
i.e. in particular for the $j$ such that $\iota_{h}(\tilde\tc,j)\neq n$).
To see why~\eqref{eq:impl} is valid, notice that
in view of~\eqref{eq:iter}, what~\eqref{eq:impl} says is that we cannot have simultaneously  both 
$\iota_{h}(\tilde\tc,j) = \overline{n}$ and
$\iota_{h}(\tilde\tc,j) = \bar m$, for some 
$m\neq n$:
indeed it is so by 
the definition of a function.

It remains to prove (iii); in view of~\eqref{eq:eq} just shown, we need to check that  
$$wr(\lambda j.\,F(\tc, \ta,\bar n, j),\kappa(\tu^*(\tilde\tc, \overline{n})), \tilde \ttb(\tu^*(\tilde\tc, \overline{n}),
\ta(\kappa(\tu^*(\tilde\tc,\overline{n})))))
$$ is the same
 as $\lambda j.\;F(\tc, \ta,\overline{n+1}, j)$.
For every $h=1,\dots, s$, this is split into three cases, corresponding to the validity check for  the three implications:
\begin{eqnarray*}
&
~~~~
i_h(\tilde\tc,j)<\bar n \to  wr(\lambda j.\,F_h(\tc, \ta,\bar n, j),\kappa_h(\tu^*(\tilde\tc, \overline{n})), \tilde t_h)(j)=F_h(\tc, \ta,\overline{n+1}, j)
\\ &
 ~~~~i_h(\tilde\tc,j)=\overline{n} \to  wr(\lambda j.\,F_h(\tc, \ta,\bar n, j),\kappa_h(\tu^*(\tilde\tc, \overline{n})), \tilde t_h)(j)=F_h(\tc, \ta,\overline{n+1}, j)
\\ &
 ~~~~i_h(\tilde\tc,j)>\overline{n} \to  wr(\lambda j.\,F_h(\tc, \ta,\bar n, j),\kappa_h(\tu^*(\tilde\tc, \overline{n})), \tilde t_h)(j)=F_h(\tc, \ta,\overline{n+1}, j)
\end{eqnarray*}
where we wrote simply $\tilde t_h$ instead of $\tilde t_h(\tu^*(\tilde\tc, \overline{n}),\ta(\kappa(\tu^*(\tilde\tc,\overline{n}))))$.
However, keeping in mind~\eqref{eq:impl} and~\eqref{eq:iter1},
the three implications can be rewritten as follows (the second one is split into two subcases) 
\begin{eqnarray*}
&
i_h(\tilde\tc,j)<\bar n \to \,F_h(\tc, \ta,\bar n, j)=F_h(\tc, \ta,\overline{n+1}, j)
\\ &
 i_h(\tilde\tc,j)=\overline{n} \wedge j= \kappa_h(\tu^*(\tc,\iota_{h}(\tilde\tc,j))) \to  \tilde t_h=F_h(\tc, \ta,\overline{n+1}, j)
\\ &
 i_h(\tilde\tc,j)=\overline{n} \wedge j\neq  \kappa_h(\tu^*(\tc,\iota_{h}(\tilde\tc,j))) \to \,F_h(\tc, \ta,\bar n, j)=F_h(\tc, \ta,\overline{n+1}, j)
\\ &
 i_h(\tilde\tc,j)>\overline{n} \to \,F_h(\tc, \ta,\bar n, j)=F_h(\tc, \ta,\overline{n+1}, j)
\end{eqnarray*}
The above four implications all hold by the definitions~\eqref{eq:transition_upd} of the $F_h$. 

\emph{Proof of Theorem~\ref{thm:tau+}} (continued).
As a consequence of the Claim, since  the formula 
$$\bigwedge_{0\leq k< n}\tilde\phi_L(\tu^*(\tilde\tc,\bar k),\td, \ta(\kappa(\tu^*(\tilde\tc,\bar k))), \ta[\td])$$
is equivalent to $\forall z~ (0\leq z < \bar n \to \tilde\phi_L(\tu^*(\tilde\tc,z),\td, \ta(\kappa(\tu^*(\tc,z))), \ta[\td]),$
we can use~\eqref{eq:p1} to express $\tau^+$ as
\begin{equation}\label{eq:tauaccelerated}
  \exists y>0
  \left(
   \begin{split}
  \forall z~ (0\leq\! z\! < y \!\to\! \tilde\phi_L(\tu^*(\tilde\tc,z),\td, \ta(\kappa(\tu^*(\tc,z))), \ta[\td])
  \wedge  \td'=\td\,\wedge
  \\
  \wedge~pc=l ~\wedge~pc'=l~
  \wedge~ \tilde\tc'=\tu^*(\tilde\tc, y)~\wedge~\ta'= \lambda j.\; F(\tc, \ta,y, j)
  \end{split}
  \right)
\end{equation}
The latter shows that $\tau^+$ is a $\Sigma^0_2$-assignment, as desired.
$\hfill\dashv$

\end{proof}

\section{A worked out example}\label{sec:all_diff_appendix}

Simple assignements might not be sufficient for nested loops where an array is scanned by a couple of counters, one of which is kept fixed (think for instance of inner loops of sorting algorithms).
To cope with these  more complicated cases, we introduce  a larger class of assignments (these assignments are  still local, hence covered by Theorem~\ref{thm:tau+}).
We call \emph{simple+} the ground assignments of the form 
\begin{equation}\label{eq:simple+}
 pc=l\;\wedge\; \phi_L({\tt c}, \td, \ta) \;\wedge\; pc'=l \wedge {\tt c}'={\tt c}\pm 1\,\wedge\,\td'=\td\;\wedge\;
  \ta'=wr(\ta, {\tt c}, \ttb({\tt c},\td,\ta)) 
 \end{equation}
where (i) $\td=d_1, \dots, d_l$ is a tuple of integer constants, 
(ii) the first occurrence of ${\tt c}$  in $wr(\ta, {\tt c}, \ttb({\tt c},\td, \ta))$ stands for  a
tuple of terms all identical to ${\tt c}$, (iii) the guard $\phi_L$ contains the conjuncts ${\tt c} \neq d_i$ ($1\leq i\leq l$), and (iv)  $\phi_L, \ttb$  are purely arithmetical 
over ${\tt c}, \td, a_1[{\tt c}], \dots a_s[{\tt c}], a_1[d_1], \dots, a_s[d_l]$.
Basically, simple+ local ground assignments differ from plain simple ones just because there are some `idle' indices $\td$; in addition, the counter ${\tt c}$ can also be decremented. 

The accelerated transition   for~\eqref{eq:simple+} computed by Theorem~\ref{thm:tau+}
can be re-written as follows (we write $j\in [{\tt c}, {\tt c}\pm k]$ for ${\tt c}\leq j\leq {\tt c}+ k$ or ${\tt c}-k\leq j\leq {\tt c}$, depending on whether we have increment or decrement
in~\eqref{eq:simple+}):
\begin{equation}
  \label{eq:transition_simple+}
 \exists k\left( 
\begin{split}
k> 0~\wedge~
 pc=l ~\wedge~
  \forall j~(j\in [{\tt c}, {\tt c}\pm k] \to\phi_L(j,\td, \ta))~\wedge~ pc'=l ~\wedge~ \td'=\td ~\wedge
  \\
\wedge~ {\tt c}'={\tt c}\pm k~\wedge~
  \ta'=\lambda j.\;(\mathtt{if}~ j\in [{\tt c}, {\tt c}\pm k]~\mathtt{then} ~\ttb(j,\td,\ta)~\mathtt{else}~\ta[j]) 
 \end{split}
 \right)
\end{equation}

To show how acceleration and abstraction/refinement techniques can
mutually benefit from each other, consider the procedure {\tt
allDiff}, represented by the {\it all diff 2} entry in \tablename
\ref{tab:experiments}. This function tests whether all entries of the
array ${\tt a}$ are pairwise different:
$$
\begin{aligned}
&  {\sf function}~{\tt allDiff}~(~{\sf int}~{\tt a}[{\tt N}]~):~\\
& 1~~{\tt r}={\sf true}; \\
& 2~~{\sf for}~( {\tt i} = 1;~{\tt i} < {\tt N} \wedge {\tt r}; {\tt i}\text{+}\text{+} )~\\
& 3~~~~~{\sf for}~({\tt j} = {\tt i}\text{-}1; {\tt j} \geq 0 \wedge {\tt r}; {\tt j}\text{-}\text{-})~\\
& 4~~~~~~~{\sf if}~({\tt a}[{\tt i}] = {\tt a}[{\tt j}])~{\tt r} = {\sf false};\\
& 5~~{\sf assert}~\left( {\tt r} \rightarrow \left( \forall x,y (0 \leq x < y < {\tt N}) \rightarrow ({\tt a}[x] \neq {\tt a}[y]) \right) \right) \\
\end{aligned}
$$
This function is represented by the transition system specified below (in the specification, we
omit identical updates to improve readability).
$$
I(\tv) = \left(r = 0 \wedge {\tt i} = 1 \wedge {\tt j} = 0 \wedge pc = l_1 \right)
$$
$$
U(\tv) = pc = l_4 \wedge \exists x, y. \left( 0 \leq x < y < {\tt a.Length}
\wedge a[x] = a[y] \right)
$$
$$
\begin{aligned}
& \tau_1 = \left(
\begin{aligned}
& pc = l_1 \wedge {\tt i} \leq {\tt a.Length} \wedge r = 0 ~\wedge\\
& pc' = l_2 \wedge
  {\tt j}' = {\tt i} - 1
\end{aligned}
\right) \\
& \tau_2 = \left(
\begin{aligned}
& pc = l_1 \wedge {\tt i} > {\tt a.Length} ~\wedge\\
& pc' = l_4
\end{aligned}
\right) \\
& \tau_3 = \left(
\begin{aligned}
& pc = l_1 \wedge r = 1 ~\wedge\\
& pc' = l_4
\end{aligned}
\right) \\
& \tau_4 = \left(
\begin{aligned}
& pc = l_3 ~\wedge\\
& pc' = l_1 \wedge
  {\tt i}' = {\tt i} + 1
\end{aligned}
\right) \\
& \tau_5 = \left(
\begin{aligned}
& pc = l_2 \wedge r = 0 \wedge {\tt j} \geq 0 \wedge a[{\tt i}] = a[{\tt j}] ~\wedge\\
& pc' = l_2 \wedge
  {\tt j}' = {\tt j} - 1 \wedge
  r = 1
\end{aligned}
\right) \\
& \tau_6 = \left(
\begin{aligned}
& pc = l_2 \wedge r = 0 \wedge {\tt j} \geq 0 \wedge a[{\tt i}] \neq a[{\tt j}] ~\wedge\\
& pc' = l_2 \wedge
  {\tt j}' = {\tt j} - 1
\end{aligned}
\right) \\
& \tau_7 = \left(
\begin{aligned}
& pc = l_2 \wedge j < 0 ~\wedge\\
& pc' = l_3
\end{aligned}
\right) \\
& \tau_8 = \left(
\begin{aligned}
& pc = l_2 \wedge r = 1 ~\wedge\\
& pc' = l_3
\end{aligned}
\right)
\end{aligned}
$$
For this problem, the transition we want to accelerate  is
$\tau_6$.
Accelerating transition $\tau_6$ is not
sufficient to avoid divergence caused by the outer loop, though. On the
other side, accelerating the inner loop simplifies the problem, which can be
successfully verified by the model checker by exploiting abstraction/refinement techniques in 1.36 seconds (see \tablename \ref{tab:experiments} for more details).

The acceleration of transition $\tau_6$ requires simple+-assignements 
(implemented in the current release of \textsc{mcmt}). We follow \textsc{mcmt} implementation quite closely to explain what happens. 

As a first observation, \textsc{mcmt} specification language requires that whenever two counters  ${\tt i}$ and ${\tt j}$ both occur in 
array applications $a[{\tt i}], a[{\tt j}]$ (like in $\tau_6$ above),
the guard of the transition must contain either the literal ${\tt i}={\tt j}$ or the literal ${\tt i}\neq {\tt j}$. Thus such transitions must be duplicated; in our case,
the copy of $\tau_6$ with ${\tt i}={\tt j}$ can be ignored because it has an inconsistent guard. The copy with ${\tt i} {\mathtt \neq} {\tt j}$ in the guard satisfies the conditions for being a simple+-assignment.
Thus, its acceleration, according to~\eqref{eq:transition_simple+}, can be written as
\begin{equation*}
 \exists k\left( 
\begin{split} 
  k> 0~\wedge ~\forall j~(j\in [{\tt j}, {\tt j}\pm k] \to  {\tt i}\neq j \wedge r = 0 \wedge j \geq 0 \wedge a[{\tt i}] \neq a[j] )
~\wedge~
\\ 
 \wedge~ pc=2 ~\wedge pc'=2 ~\wedge~ {\tt i}'={\tt i}~\wedge~ r'=r ~\wedge ~ {\tt j}'={\tt j}\pm k~\wedge  ~\ta'=\ta 
 \end{split}
 \right)
\end{equation*}
In the current release, \textsc{mcmt} is able to compute by itself the above accelerated transition and thus to certify safety of {\tt
allDiff} procedure.

\section{Experimental evaluation}\label{sec:experiments_appendix}
Complete statistics for the experiments performed with {\sc mcmt} are
reported in \tablename \ref{tab:experiments}.
Benchmarks have been taken from different sources:
\begin{itemize}
\item The benchmarks ``filter test'', ``max in array test'', ``filter'', ``max in array 1'', ``max in
array 2'', ``max in array 3'' have been taken and/or adapted from
programs on \url{http://proval.lri.fr/}.
\item The ``heap as array'' program has been suggested by K. Rustan M.
Leino and it is reported in \figurename \ref{fig:heap_program}.
\item all the programs $pN$ have been taken from ``I. Dillig, T. Dillig,
and A. Aiken. Fluid updates: Beyond strong vs. weak updates. In ESOP,
pages 246-266, 2010.''.
\item The ``bubble sort'' example comes from the ``Eureka'' project
\url{http://www.ai-lab.it/eureka} and has been used as a benchmark in the paper
``A. Armando, M. Benerecetti, and J. Mantovani. Abstraction refinement
of linear programs with arrays. In TACAS, pages 373-388, 2007.''
\item ``all diff 1'' and ``all diff 2'' have been suggested by Madhusudan
Parthasarath and his group. They represent two different encoding of an
algorithm that initializes an array to different values and then check
if the array has been correctly initialized.
\item ``compare'', ``copy'', ``find 1'', ``find 2'', ``init'', ``init test'', ``partition'' have been taken/adapted from ``Krystof Hoder, Laura Kov\'acs, Andrei Voronkov: Interpolation and Symbol Elimination in Vampire. In IJCAR, pages 188-195, 2010''.
\item The ``linear search'' program is used as a running example on the
book ``Aaron R. Bradley, Zohar Manna: The calculus of computation -
decision procedures with applications to verification. Springer 2007,
pp. I-XV, 1-366''.
\item ``selection sort'' example has been used in ``M. N. Seghir, A.
Podelski, and T. Wies. Abstraction Refinement for Quantified
Array Assertions. In SAS, pages 3-18, 2009.''
\item ``strcmp'', ``strcpy'' and ``strlen'' have been adapted from the
standard string C library.
\end{itemize}
The benchmarks named with `` *
test '' refer to benchmarks with quantified assertions substituted by a for loop. For those
programs, the postcondition does not have quantifiers: in these benchmarks it is even harder to come up with a quantified
safe inductive invariant to prove that the program is correct. 
Thus, it is a remarkable fact that our tool can automatically synthetize such invariants. 
\begin{table}\scriptsize
\centering
\begin{tabular}{|l|l|l|l|l|l|} \hline
~{\sc Program}     ~~& ~{\sc Status} ~~&~ {\sc No options} ~~&~ {\sc Abstraction} ~~&~ {\sc Acceleration} ~~&~ {\sc Accel. + Abstr.} \\ \hline \hline
~filter test       ~~& ~safe         ~~& ~$\times$           & ~0.08                & ~$\times$             & ~0.08                  \\
~heap as array     ~~& ~safe         ~~& ~$\times$           & ~0.12                & ~$\times$             & ~0.12                  \\
~init test         ~~& ~safe         ~~& ~$\times$           & ~11.72               & ~$\times$             & ~0.16                  \\
~max in array test ~~& ~safe         ~~& ~$\times$           & ~0.18                & ~$\times$             & ~0.18                  \\
~p01               ~~& ~safe         ~~& ~$\times$           & ~$\times$            & ~0.09                 & ~9.08                  \\
~p02               ~~& ~safe         ~~& ~$\times$           & ~$\times$            & ~0.09                 & ~9.52                  \\
~p03               ~~& ~safe         ~~& ~$\times$           & ~0.11                & ~0.09                 & ~0.14                  \\
~p08               ~~& ~safe         ~~& ~$\times$           & ~0.12                & ~0.12                 & ~0.11                  \\
~p09               ~~& ~safe         ~~& ~$\times$           & ~0.12                & ~0.99                 & ~0.11                  \\
~p14               ~~& ~safe         ~~& ~$\times$           & ~6.39                & ~0.35                 & ~7.78                  \\
~p17               ~~& ~safe         ~~& ~$\times$           & ~0.02                & ~0.19                 & ~0.19                  \\
~p04               ~~& ~unsafe       ~~& ~0.02               & ~0.03                & ~0.03                 & ~0.02                  \\
~p10               ~~& ~unsafe       ~~& ~0.07               & ~0.04                & ~0.06                 & ~0.03                  \\
~p11               ~~& ~unsafe       ~~& ~0.02               & ~0.03                & ~0.04                 & ~0.04                  \\
~p15               ~~& ~unsafe       ~~& ~1.4                & ~1.74                & ~0.3                  & ~2.97                  \\
~p16               ~~& ~unsafe       ~~& ~4.27               & ~3.70                & ~0.45                 & ~8.89                  \\
~p18               ~~& ~unsafe       ~~& ~0.01               & ~0.02                & ~0.01                 & ~0.01                  \\
~p19               ~~& ~unsafe       ~~& ~0.02               & ~0.02                & ~0.01                 & ~0.01                  \\
~p20               ~~& ~unsafe       ~~& ~0.02               & ~0.02                & ~0.03                 & ~0.02                  \\
~p22               ~~& ~unsafe       ~~& ~0.02               & ~0.03                & ~0.02                 & ~0.17                  \\ \hline 
~all diff 1        ~~& ~safe         ~~& ~$\times$           & ~$\times$            & ~0.08                 & ~0.13                  \\
~all diff 2        ~~& ~safe         ~~& ~$\times$           & ~$\times$            & ~$\times$             & ~1.36                  \\
~bubble sort       ~~& ~safe         ~~& ~$\times$           & ~1.23                & ~$\times$             & ~1.23                  \\
~compare           ~~& ~safe         ~~& ~$\times$           & ~0.04                & ~$\times$             & ~0.04                  \\
~copy              ~~& ~safe         ~~& ~$\times$           & ~0.03                & ~0.03                 & ~0.03                  \\
~filter            ~~& ~safe         ~~& ~$\times$           & ~0.11                & ~$\times$             & ~0.11                  \\
~find 1            ~~& ~safe         ~~& ~$\times$           & ~0.06                & ~$\times$             & ~0.06                  \\
~find 2            ~~& ~safe         ~~& ~$\times$           & ~0.07                & ~0.06                 & ~0.17                  \\
~init              ~~& ~safe         ~~& ~$\times$           & ~0.08                & ~0.03                 & ~0.1                   \\
~linear search     ~~& ~safe         ~~& ~$\times$           & ~0.04                & ~0.05                 & ~0.02                  \\
~max in array 1    ~~& ~safe         ~~& ~$\times$           & ~0.1                 & ~$\times$             & ~0.1                   \\
~max in array 2    ~~& ~safe         ~~& ~$\times$           & ~0.11                & ~$\times$             & ~0.13                  \\
~max in array 3    ~~& ~safe         ~~& ~$\times$           & ~0.06                & ~$\times$             & ~0.01                  \\
~minusN            ~~& ~safe         ~~& ~$\times$           & ~$\times$            & ~0.77                 & ~1.4                   \\
~partition         ~~& ~safe         ~~& ~$\times$           & ~0.05                & ~$\times$             & ~0.03                  \\
~selection sort    ~~& ~safe         ~~& ~$\times$           & ~7.87                & ~$\times$             & ~45.07                 \\
~strcat 1          ~~& ~safe         ~~& ~$\times$           & ~$\times$            & ~$\times$             & ~3.5                   \\
~strcat 2          ~~& ~safe         ~~& ~$\times$           & ~$\times$            & ~$\times$             & ~3.62                  \\
~strcmp            ~~& ~safe         ~~& ~$\times$           & ~0.04                & ~0.06                 & ~0.02                  \\
~strcpy            ~~& ~safe         ~~& ~$\times$           & ~0.03                & ~0.02                 & ~0.01                  \\
~strlen            ~~& ~safe         ~~& ~$\times$           & ~$\times$            & ~0.1                  & ~0.06                  \\
~p01               ~~& ~safe         ~~& ~$\times$           & ~0.08                & ~0.02                 & ~0.1                   \\
~p02               ~~& ~safe         ~~& ~$\times$           & ~0.08                & ~0.05                 & ~0.1                   \\  
~p03               ~~& ~safe         ~~& ~$\times$           & ~0.03                & ~0.02                 & ~0.03                  \\
~p08               ~~& ~safe         ~~& ~$\times$           & ~0.03                & ~0.05                 & ~0.03                  \\
~p09               ~~& ~safe         ~~& ~$\times$           & ~0.03                & ~0.04                 & ~0.03                  \\
~p18               ~~& ~safe         ~~& ~$\times$           & ~$\times$            & ~0.07                 & ~0.33                  \\
~p20               ~~& ~safe         ~~& ~$\times$           & ~0.04                & ~0.05                 & ~0.02                  \\ 
~p04               ~~& ~unsafe       ~~& ~0.07               & ~0.02                & ~0.01                 & ~0.01                  \\
~p11               ~~& ~unsafe       ~~& ~0.01               & ~0.02                & ~0.02                 & ~0.01                  \\
~p14               ~~& ~unsafe       ~~& ~0.31               & ~1.79                & ~0.28                 & ~2.5                   \\
~p15               ~~& ~unsafe       ~~& ~0.09               & ~1.77                & ~0.12                 & ~1.4                   \\
~p16               ~~& ~unsafe       ~~& ~0.11               & ~2.97                & ~1.23                 & ~6.57                  \\
~p17               ~~& ~unsafe       ~~& ~0.02               & ~0.03                & ~0.01                 & ~0.02                  \\
~p19               ~~& ~unsafe       ~~& ~0.02               & ~0.02                & ~0.01                 & ~0.01                  \\ \hline
\end{tabular}
\caption{\label{tab:experiments}Experimental results for different
options. Time limit has been set to $60$ seconds, and $\times$ denotes a
timeout. Programs in the first part of the table are annotated with
quantifier-free assertions, those in the second part have
$\forall$-assertions. Notably, when abstraction and acceleration is
combined {\sc mcmt} is able to verify all the 55 programs.}
\end{table}
\begin{figure}[t]
{\footnotesize
\begin{tabbing}
 assssssssssssssssssssssssss \= s \= s \= s \= s \= s \kill
 \> {\sf var} ${\tt Heap}$: [{\sf int}] {\sf int}; \\
 \> {\sf const} {\sf unique} ${\tt F}$: {\sf int}; ~~{\sf const} {\sf unique} ${\tt G}$: {\sf int}; \\
 \> {\sf const} ${\tt F\_final}$: {\sf int}; ~~{\sf const} ${\tt G\_final}$: {\sf int}; \\
 \> {\sf procedure} {\sf HeapP} ( ) \\
 \> \>  {\sf modifies} ${\tt Heap}$; \\
 \> \>  {\sf requires} ${\tt F\_final} > 0 \wedge {\tt G\_final} > 0$; \\
 \> \>  {\sf ensures} ${\tt Heap}[{\tt F}] = {\tt F\_final} \wedge {\tt Heap}[{\tt G}] = {\tt G\_final}$; \\ 
 \> \{ \\
 \> \>  ${\tt Heap}[{\tt F}]$ := 0; ~ ${\tt Heap}[{\tt G}]$ := ${\tt G\_final}$; \\ 
 \> \>  {\sf while} (${\tt Heap}[{\tt F}] < {\tt F\_final}$) \\
 \> \> \>   {\sf invariant} ${\tt Heap}[{\tt F}] \leq {\tt F\_final}$; \\
 \> \>  \{ \\
 \> \> \>   ${\tt Heap}[{\tt F}]$ := ${\tt Heap}[{\tt F}] + 1$;\\
 \> \>  \} \\
 \> \}
\end{tabbing}
}
\caption{\label{fig:heap_program}The ``heap as array'' program.}
\end{figure}


\end{document}